\documentclass[preprint,10pt,3p,authoryear]{elsarticle}

\usepackage{latexsym}
\usepackage{amsmath, amsfonts, amsthm,amssymb}
\usepackage{epsfig}
\usepackage{stmaryrd}
\usepackage{multicol}
\usepackage{pdfsync}
\usepackage{bm}
\usepackage{dsfont}
\usepackage{comment}
\usepackage{algorithm}
\usepackage{algpseudocode}
\usepackage{caption}
\usepackage{psfrag}
\usepackage{subfig}
\usepackage{xcolor}
\usepackage{graphics}
\usepackage{enumitem}
\usepackage{mathtools}
\usepackage{bbm}
\usepackage{mathrsfs} 
\usepackage{color}
\usepackage{hyperref}
\usepackage{cancel}
\usepackage{float}
\usepackage{todonotes}
\usepackage{setspace}
\usepackage{booktabs}
\usepackage{subcaption}
\usepackage{graphicx}

\usepackage[utf8]{inputenc}
\usepackage{titlesec}
\titleformat{\subsection}
{\normalfont\fontsize{10}{17}\bfseries}{\thesubsection}{1em}{}

\usepackage{tikz-cd}
\journal{TBA}
\hypersetup{
	colorlinks = true, 
	urlcolor = blue, 
	linkcolor = red, 
	citecolor = blue 
}

\newtheorem{theorem}{Theorem}

\newtheorem{corollary}{Corollary}
\newtheorem{lemma}{Lemma}
\newtheorem{proposition}{Proposition}
\newtheorem{definition}{Definition}

\theoremstyle{remark}
\newtheorem{remark}[theorem]{Remark}

\numberwithin{equation}{section}

\newcommand{\Tc}{{\mathcal T}}
\renewcommand{\d}{\mathrm{d}}
\DeclareMathSymbol{\shortminus}{\mathbin}{AMSa}{"39}

\title{
Optimal Exit Time for Liquidity Providers in Automated Market Makers
}

\author[label3]{Philippe Bergault}
\author[label3]{Sébastien Bieber}
 \author[label1,label2]{Leandro S\'{a}nchez-Betancourt}

\address[label3]{Ceremade, Université Paris Dauphine-PSL}
\address[label1]{Mathematical Institute, University of Oxford}
\address[label2]{Oxford-Man Institute of Quantitative Finance}

\begin{document}

\newcommand{\la}{\left \langle}
\newcommand{\ra}{\right\rangle}
\newcommand{\cb}[1]{{\color{blue} #1}}
\newcommand{\norm}[1]{\left\lVert #1 \right\rVert}
\newcommand{\bae}{\begin{equation}\begin{aligned}}
\newcommand{\eae}{\end{aligned}\end{equation}}
\newcommand{\beq}{\begin{equation}}
\newcommand{\eeq}{\end{equation}}
\newcommand{\N}{\mathbb{N}}
\newcommand{\R}{\mathbb{R}}
\newcommand{\E}{\mathbb{E}}
\newcommand{\Pb}{\mathbb{P}}
\newcommand{\PbI}{\mathbb{P}^I}
\newcommand{\PbB}{\mathbb{P}^B}
\newcommand{\Lb}{\mathbb{L}}

\newcommand{\mfT}{{\mathfrak{T}}}
\newcommand{\mfO}{{\mathfrak{O}}}
\newcommand{\tT}{{t\in\mfT}}
\newcommand{\mcA}{{\mathcal{A}}}
\newcommand{\mcC}{{\mathcal{C}}}
\newcommand{\mcF}{{\mathcal{F}}}
\newcommand{\mcN}{{\mathcal{N}}}
\newcommand{\mcB}{{\mathcal{B}}}
\newcommand{\mcH}{{\mathcal{H}}}

\newcommand{\transB}{\mathfrak{t}}
\newcommand{\decayB}{\mathfrak{p}}
\newcommand{\instantB}{\mathfrak{h}}

\newcommand{\assign}{:=}
\newcommand{\nobracket}{}
\newcommand{\tmop}[1]{\ensuremath{\operatorname{#1}}}
\newcommand{\tmtextit}[1]{\text{{\itshape{#1}}}}
%

\newcommand{\F}{\mathcal{F}}
\newcommand{\Fb}{\mathbb{F}}
\newcommand{\Prob}{\mathbb{P}}
\newcommand{\X}{\mathbb{X}}
\newcommand{\Ss}{\mathcal{S}}

\newcommand{\Real}{\mathbb{R}}
\newcommand{\Y}{\mathbb{\mathcal{Y}}}
\newcommand{\tmL}{\mathbb{L}}
\newcommand{\Ll}{\mathcal{L}}
\newcommand{\Pp}{\mathcal{P}}
\newcommand{\Bb}{\mathcal{B}}
\newcommand{\Dd}{\mathcal{D}}
\newcommand{\Ee}{\mathcal{E}}
\newcommand{\Nn}{\mathcal{N}}
\newcommand{\Q}{\mathbb{Q}}
\newcommand{\ch}{\mathds{1}}

\newcommand{\wealth}{\text{M}}
\newcommand{\rwealth}{Q}
\newcommand{\rwealthD}{q}

\newcommand{\tempLP}{\mathfrak{a}}
\newcommand{\tempB}{\mathfrak{a}}
\newcommand{\tempI}{\mathfrak{b}}
\newcommand{\tempU}{\mathfrak{c}}
\newcommand{\permB}{\mathfrak{p}}
\newcommand{\termpB}{\phi}
\newcommand{\termpI}{\psi}
\newcommand{\runnB}{r^B}
\newcommand{\runnI}{r^I}

\renewcommand{\d}{\mathrm{d}}

\newcommand{\stateB}{\mathfrak{y}}
\newcommand{\stateI}{\mathfrak{x}}
\newcommand{\state}{\mathfrak{y}}
\newcommand{\statevn}{\mathfrak{y}}

\newcommand{\nuI}{\eta}
\newcommand{\nuB}{\nu}
\newcommand{\nuU}{\xi}

\newcommand{\nuBstar}{\nu^*}
\newcommand{\nuIstar}{\eta^*}

\newcommand{\varI}{\mathbb{V}^I}
\newcommand{\varB}[1][B]{\mathbb{V}^{#1}}

\newcommand{\betaI}{\beta_0^I}
\newcommand{\betaaI}{\beta_1^I}
\newcommand{\rhoI}{\rho_0^I}
\newcommand{\rhooI}{\rho_1^I}

\newcommand{\betaB}{\beta_0^B}
\newcommand{\betaaB}{\beta_1^B}
\newcommand{\rhoB}{\rho_0^B}
\newcommand{\rhooB}{\rho_1^B}

\newcommand{\runcostI}{\rho_0^I + \rho_1^I\, \varI}
\newcommand{\runcostB}{\rho_0^B + \rho_1^B\, \varB}

\newcommand{\termcostI}{\beta_0^I + \beta_1^I\, \varI}
\newcommand{\termcostB}{\beta_0^B + \beta_1^B\, \varB}

\newcommand{\constrB}{\tempB - f_2^2\,\tempI}

\newcommand{\MI}[1][I]{M^{#1}}
\newcommand{\MB}[1][B]{M^{#1}}
\newcommand{\MBa}[1][B]{\tilde{N}^{#1}}
\newcommand{\MBb}[1][B]{\tilde{M}^{#1}}
\newcommand{\MZ}[1][Z]{M^{#1}}

\newcommand{\QI}[1][I]{Q^{#1}}
\newcommand{\QB}[1][B]{Q^{#1}}
\newcommand{\QU}[1][U]{Q^{#1}}
\newcommand{\XI}[1][I]{X^{#1}}
\newcommand{\XB}[1][B]{X^{#1}}
\newcommand{\XU}[1][U]{X^{#1}}

\newcommand{\QIstar}[1][I]{Q^{#1,*}}
\newcommand{\QBstar}[1][B]{Q^{#1, *}}
\newcommand{\XIstar}[1][I]{X^{#1, *}}
\newcommand{\XBstar}[1][B]{X^{#1, *}}

\newcommand{\mcFB}{\mcF^B}
\newcommand{\mcFI}{\mcF^I}
\newcommand{\vfB}{J^{B*}}
\newcommand{\vfI}{J^{I*}}
\newcommand{\pcB}{J^B}
\newcommand{\pcI}{J^I}
\newcommand{\mcAB}{\mcA^B}
\newcommand{\mcAI}{\mcA^I}

\newcommand{\g}{g}
\newcommand{\gI}{g^{I}}
\newcommand{\gB}{g^{B}}
\newcommand{\gZ}{g^{Z}}
\newcommand{\gY}{g^{Y}}
\newcommand{\h}{h}
\newcommand{\hI}{h^{I}}
\newcommand{\hB}{h^{B}}
\newcommand{\hZ}{h^{Z}}
\newcommand{\hY}{h^{Y}}
\newcommand{\f}{f}
\newcommand{\fI}{f^{I}}
\newcommand{\fB}{f^{B}}
\newcommand{\fZ}{f^{Z}}
\newcommand{\fY}{f^{Y}}

\newcommand{\Pn}{P^{\hat{\nu}}}
\newcommand{\Pa}{P^{\hat{\alpha}}}

\newcommand{\sigmaU}{\sigma^U}
\newcommand{\sigmaM}{\sigma^M}
\newcommand{\sigmaS}{\sigma^S}

\newcommand{\vp}{\varphi}

\newcommand{\rew}{\mathfrak{R}}
\newcommand{\fee}{\mathfrak{r}}
\newcommand{\feeV}{\mathfrak{r}}
\newcommand{\cost}{\mathfrak{c}}

\newcommand{\Na}{N^-}
\newcommand{\Nb}{N^+}
\newcommand{\Nab}{N^\pm}
\newcommand{\Ni}{N^i}
\newcommand{\hatNa}{\hat{N}^a}
\newcommand{\hatNb}{\hat{N}^b}
\newcommand{\hatNab}{\hat{N}^{b/a}}
\newcommand{\hatNi}{\hat{N}^i}
\newcommand{\tildeNa}{\tilde{N}^a}
\newcommand{\tildeNb}{\tilde{N}^b}
\newcommand{\tildeNab}{\tilde{N}^{b/a}}
\newcommand{\tildeNi}{\tilde{N}^i}

\newcommand{\lambdaa}{\lambda^a}
\newcommand{\lambdab}{\lambda^b}
\newcommand{\lambdaab}{\lambda^{b/a}}
\newcommand{\lambdai}{\lambda^i}
\newcommand{\barlambdaa}{\bar{\lambda}^a}
\newcommand{\barlambdab}{\bar{\lambda}^b}
\newcommand{\barlambdaab}{\bar{\lambda}^{b/a}}
\newcommand{\barlambdai}{\bar{\lambda}^i}

\newcommand{\Aa}{A^-}
\newcommand{\Ab}{A^+}
\newcommand{\Aab}{A^\pm}
\newcommand{\Ai}{A^i}
\newcommand{\AW}{A^W}
\newcommand{\AB}{A^B}

\newcommand{\Deltarwealtha}{\Delta^-} 
\newcommand{\Deltarwealthb}{\Delta^+}
\newcommand{\Deltarwealthab}{\Delta^\pm}
\newcommand{\Deltarwealthi}{\Delta^i}

\newcommand{\deltaa}{\delta^-}
\newcommand{\deltab}{\delta^+}
\newcommand{\deltai}{\delta^i}
\newcommand{\deltaab}{\delta^\pm}

\newcommand{\Za}{Z^-}
\newcommand{\Zb}{Z^+}
\newcommand{\Zi}{Z^i}
\newcommand{\Zab}{Z^\pm}

\newcommand{\va}{v^-}
\newcommand{\vb}{v^+}
\newcommand{\vi}{v^i}
\newcommand{\vab}{v^\pm}

\newcommand{\Vvn}{V}  
\newcommand{\Vvna}{V^-}
\newcommand{\Vvnb}{V^+}
\newcommand{\Vvni}{V^i}
\setlength\parindent{0pt}

\begin{abstract}
We study the problem of optimal liquidity withdrawal for a representative liquidity provider (LP) in an automated market maker (AMM). LPs earn fees from trading activity but are exposed to impermanent loss (IL) due to price fluctuations. While existing work has  focused on static provision and exogenous exit strategies, we characterise the optimal exit time as the solution to a stochastic control problem with an endogenous stopping time. Mathematically, the LP's value function is shown to satisfy a Hamilton–Jacobi–Bellman quasi-variational inequality, for which we establish uniqueness in the viscosity sense. To solve the problem numerically, we develop two complementary approaches: a Euler scheme based on operator splitting and a Longstaff–Schwartz regression method. Calibrated simulations highlight how the LP's optimal exit strategy depends on the oracle price volatility, fee levels, and the behaviour of arbitrageurs and noise traders. Our results show that while arbitrage generates both fees and IL, the LP's optimal decision balances these opposing effects based on the pool state variables and price misalignments. 
Lastly, we find the optimal fee level for the representative LP when they play  the exit strategy we derived.
This work contributes to a deeper understanding of dynamic liquidity provision in AMMs and provides insights into the sustainability of passive LP strategies under different market regimes.

\vspace{0.5cm}

\end{abstract}

\maketitle

\tableofcontents

\section{Introduction}

Decentralised finance (DeFi) is evolving rapidly. Since the launch of Uniswap, there is a flurry of innovations  in the DeFi space, many of which revolve around the so-called automated market maker (AMM) technology. AMMs are decentralised trading venues where liquidity providers (LPs) and liquidity takers (LTs) interact under pre-defined trading rules; see \cite{adams2021uniswap,adams2023uniswap,capponi2021adoption}. AMMs such as Uniswap are key to DeFi, enabling users to trade assets without using traditional limit order books. Among the most prominent AMM designs, the constant function market maker (CFM) maintains an invariant of the form \( f(x, y) = k \), where \( x \) and \( y \) denote the reserves of two assets, and \( f \) is a fixed function that determines the trading rules. A key attraction of CFMs is its passive liquidity provision: liquidity providers are able to supply capital to the pool without the need for continuous rebalancing or active order placement, as is required in traditional limit order book markets. In return for providing liquidity, LPs earn a pro-rata share of the fees collected on trades executed in the AMM. However, LPs also face risks, most notably the so-called impermanent loss (IL), which arises due to price fluctuations between the pooled assets and external market prices.\\

Although still in its early stages, the literature on AMMs is rapidly expanding in several key directions. A foundational contribution is provided by \cite{angeris2021analysis}, who analyse the core properties of CFMs; see also \cite{angeris2020improved} for a general multi-coin setting. In particular, they show that, in the presence of arbitrageurs, the exchange rate proposed by a CFM for small trades remains within a band around the external market price, where the width of the band is determined by the level of transaction fees. The authors also establish several fundamental properties of CFMs, such as the no-splitting and no-depletion properties, and they characterise the payoff of LPs as a function of the external asset price in the absence of fees;  see also \cite{clark2020replicating}, and its extension to concentrated liquidity in Uniswap v3 in \cite{clark2021replicating}. Generalisations of AMM designs have been proposed in various directions. \cite{cartea2024strategic} introduce decentralised liquidity pools, generalising CFMs and offering LPs a wide range of dynamic bonding curves to improve capital efficiency, while \cite{bergault2024automated, bergault2024pegged, bergault2024price} propose an oracle-based AMM architecture aimed at mitigating impermanent loss. More broadly, the problem of hedging impermanent loss is studied in \cite{fukasawa2024model}, where the authors analyse how LPs can reduce their exposure through dynamic trading strategies, and in \cite{milionis2022automated, milionis2024automated} where the authors introduce the now well-known concept of Loss-Versus-Rebalancing. The related notion of predictable loss is discussed in \cite{cartea2024decentralized}, where the authors also analyse liquidity provision and trading strategies; see also \cite{milionis2023myersonian}, and \cite{fan2021strategic}. Arbitrage between centralised and decentralised exchanges have been studied in \cite{cartea2023decentralised, he2025arbitrage}, and more generally the competition between limit order books and AMMs is discussed in \cite{aoyagi2025coexisting, barbon2021quality} and \cite{lehar2025decentralized}. Finally, equilibrium  between LPs is studied in \cite{aoyagi2020liquidity} and \cite{hasbrouck2022need}.
\\

It is now well understood that the fees paid to LPs are the main mechanism  to compensate for the IL that arises when arbitrageurs align the pool's reference price and the external market price. More precisely, in the absence of fees, an agent providing liquidity in a standard CFM is exposed to a concave payoff that is strictly dominated by holding the assets outside the pool. Optimal fee structures have been studied in \cite{evans2021optimal, cao2025structural, baggiani2025optimal, campbell2025optimal}, while the problem of incentive design for LPs is addressed in \cite{aqsha2025equilibrium} through a game-theoretic argument. The question of take rate -- the proportion
of the fees kept by the protocol -- has also been studied in \cite{fritsch2022economics}. While fees are ideally meant to offset IL, empirical and theoretical studies have shown that LPs often experience persistent negative returns, particularly in volatile markets; see e.g., \cite{canidio2023arbitrageurs}. This poses a serious challenge to the long-term sustainability of passive liquidity provision; see, e.g., \cite{hasbrouck2025economic}.\\

In this paper, we are interested in studying the optimal exit time from a CFM when the goal is to minimise IL and to maximise the fees collected. More precisely, the LP can withdraw liquidity at any time before a fixed terminal horizon~$T$, and aims to maximise the expected value of their position, accounting for both accrued fees and the realised impermanent loss coming from the  price evolution. This naturally leads to a stochastic control problem, where the exit time is treated as a control variable. Closest to our work is the recent preprint by \cite{zhu2025optimal}, which also studies the optimal exit time from a CFM.\footnote{We are also aware of other concurrent efforts within this general theme. For instance, \cite{ma2025optimal} study allocation and exit decisions in a liquid staking protocol and an AMM. Recently, Zachary Feinstein and Marina Georgiou presented related preliminary results at the 2025 SIAM Conference on Financial Mathematics and Engineering (no preprint available at the time of writing). See also \cite{agarwal2025optimal}.} This paper provides a tractable framework with  closed-form expressions that are easy to implement and to analyse. In that sense, it offers valuable structural insights into the optimal exit problem. While our work addresses a similar question, it departs from \cite{zhu2025optimal} in several key directions. First, the optimal threshold derived in the latter is static, while in our formulation it is dynamic: the exit region evolves over time and depends on the current (stochastic and controlled) state variables. Second, \cite{zhu2025optimal} assumes that ``the price'' follows a geometric Brownian motion, whereas we explicitly distinguish between a fundamental (external) price -- where price formation occurs -- and the pool's internal reference price, which we derive from a detailed model of order flows within the AMM. Third, while \cite{zhu2025optimal} abstracts from strategic liquidity takers, our framework models arbitrage activity explicitly. This allows us to analyse how arbitrageurs align internal and external prices, how their trades generate realised IL, and how this affects the optimal exit decision. In our setting, the optimal strategy depends not only on the external price (e.g., on Binance) but also on the state of the pool (e.g., reserves and internal price). Finally, we extend our analysis to account for risk aversion on the LP’s side — a natural feature in applications and, to the best of our knowledge, not yet addressed in the existing literature.\\

Our key contributions are as follows.
\begin{enumerate}
    \item We characterise the optimal exit time of a liquidity provider from a CFM. Mathematically, we show that the value function of the corresponding optimal stopping problem is the unique viscosity solution to a Hamilton–Jacobi–Bellman quasi-variational inequality (HJB QVI) that delineates the optimal exit region.

    \item We propose two numerical approaches to solve the problem: an Euler scheme based on operator splitting, and a regression-based Longstaff–Schwartz algorithm.

    \item We calibrate the model and test the optimal strategy. Our results reveal several key findings:
    \begin{enumerate}
        \item Both the average fees collected and the average impermanent loss at the optimal exit time are concave functions of the volatility of the oracle price.

        \item The performance criterion (defined as fees collected minus impermanent loss), which is non-negative in expectation since exiting the pool immediately is always admissible, increases approximately linearly with the fee level beyond a certain threshold. This is because, once fees are high enough, LPs optimally choose to remain in the pool until the terminal time.

        \item Higher activity by arbitrageurs and noise traders leads to higher fee revenues. However, arbitrageurs also increase the realised impermanent loss at the exit time, up to the maximum loss implied by the oracle price volatility. In contrast, an increase in noise trader activity raises the fees collected while leaving impermanent loss unchanged on average. 

        \item LPs tend to exit the pool when the misalignment between the oracle price and the pool’s internal reference price becomes too large. At that point, the expected gains from future arbitrage trades are outweighed by the associated impermanent loss, prompting LPs to exit before arbitrageurs align prices.
    \end{enumerate}
\end{enumerate}

The remainder of the paper proceeds as follows. Section \ref{sec: the model} introduces the probabilistic framework and formulates the optimal stopping problem. Section \ref{sec: mathematical analysis} analyses the problem using dynamic programming techniques, and in particular characterises the value function as the unique viscosity solution of an HJB QVI equation. Section \ref{sec: numerical results} studies the solution of the problem and collects insights on the structure of the optimal strategy, and Section \ref{sec: conclusions} concludes. \ref{sec:app_risk_averse} generalises our result to a risk-averse LP with exponential utility, while the proof of the main result is reported in \ref{sec:proof_viscosity}.

\section{The model}\label{sec: the model}

\subsection{Probabilistic framework}

Let $T>0$ be a trading horizon, $\Omega_c$ the set of continuous functions from $\mfT = [0,T]$ into $\mathbb R$, $\Omega_d$ the set of piecewise constant càdlàg functions from $\mfT$ into $\mathbb N$, and $\Omega = \Omega_c \times \big( \Omega_d \big)^2$ with the corresponding Borel algebra $\mathcal F$. The observable state is the canonical process $ (\chi_t)_{t \in \mfT} = \big(W_t,  \hat N^b_t,  \hat N^a_t \big)_{t \in \mfT}$ of the measurable space $(\Omega, \mathcal F)$, with
\[
 W_t(\omega) := w(t), \;  \hat N^b_t(\omega) :=   \hat n^b(t), \;  \hat N^a_t(\omega) :=   \hat n^a(t), \; \text{for all } t \in \mfT,
\]
where $\omega := (w,\hat n^b, \hat n^a) \in \Omega.$\\

We introduce a probability measure $\hat {\mathbb P}$ such that $W$ is a Brownian motions and $ \hat N^a$, $ \hat N^b$ are Poisson processes with intensity $a_0>0$. Under this probability measure, $W$, $ \hat N^a$ and $ \hat N^b$ are independent. We endow the space $(\Omega, \mathcal F )$ with the $\hat{\mathbb P}-$augmented canonical filtration $\mathbb F := \left(\mathcal F_t  \right)_{t\in \mfT} $ generated by $ (\chi_t)_{t \in \mfT}$.\\

We study an automated market maker (AMM) for a pair of assets $X$ and $Y$ (e.g., USDC and ETH) with trading function $f(x,y)=x\,y$. This type of AMM is referred to as constant function market (CFM). We assume that there exists also an external limit order book (LOB) venue for trading in $X$ and $Y$, where the price formation occurs. We let $S=(S_t)_{t\in \mfT}$ be the external mid-price of asset $Y$ in terms of asset $X$. Within the CFM, we let $X_t$ and $Y_t$ be the strictly positive quantities of assets $X$ and $Y$ sitting in the pool at time $t\in\mfT$.  We let the function $\varphi_{c}(y)$ be the level function of $f$ given by
\begin{equation}
    \varphi_{c}(y) = \frac{c}{y}\,.
\end{equation}
The parameter $c>0$ in the pool is known as the depth of the pool which we take to be constant throughout the trading window.\\

The external mid-price follows the dynamics
\begin{equation}
     S_t = S_0 + \sigma \,W_t\,,\qquad \text{with } S_0\in\mathbb{R}^+\,.
\end{equation}
In the CFM, liquidity takers arrive according to the counting processes $N^a$ (number of liquidity taking buys) and $N^b$ (number of liquidity taking sells) that model the number of trades of size $\xi>0$ through time.\footnote{We assume a constant trade size for the sake of parsimony, but more general models with random trade sizes can be considered as well within our framework.} We denote by $Z=(Z_t)_{t\in \mfT}= (X_t/Y_t)_{t\in \mfT}$ the marginal price of $Y$ in terms of $X$ in the CFM. We follow the characterization of CFMs in \cite{cartea2024strategic} that describes the mechanics of the reserves $X$, $Y$, and the instantaneous rate $Z$,  according to the arrival of orders $N^{a}$ and $N^b$.  We have that
\begin{align*}
    \d X_t &=  
    \left[ \varphi_{c}(Y_{t\shortminus} + \xi) - \varphi_{c}(Y_{t\shortminus}) \right]\, \d N^b_t + \left[ \varphi_{c}(Y_{t\shortminus} - \xi) - \varphi_{c}(Y_{t\shortminus}) \right]\, \d N^a_t 
    \, ,\\
    \d Y_t &=  
    \xi\, \d N^b_t - \xi\, \d N^a_t
    \, ,\\
    \d Z_t &= 
    \left[-(\varphi_{c})'(Y_{t\shortminus}+\xi) +  (\varphi_{c})'(Y_{t\shortminus}) \right]\, \d N^b_t + \left[-(\varphi_{c})'(Y_{t\shortminus}-\xi) +  (\varphi_{c})'(Y_{t\shortminus}) \right]\, \d N^a_t
    \, , \\
    \d N^b_t &= \mathds 1_{\{Y_{t\shortminus} +\xi \le \overline Y\}}\d \hat N^b_t \, , \\
    \d N^a_t &= \mathds 1_{\{Y_{t\shortminus} -\xi \ge \underline Y\}}\d \hat N^a_t \, ,
\end{align*}
with $N^b_0=N^a_0 =0$, $X_0 \in \mathbb{R}^+$, $Y_0\in[\underline Y, \overline Y]$, and $Z_0=X_0/Y_0$. The bounds $0< \underline Y < \overline Y <\infty$ are mainly introduced to ease the mathematical analysis. This is not an issue in practice as they can be chosen arbitrarily small and large, respectively. One can also interpret them as risk limits, i.e. the AMM stops trading when the instantaneous prices are too low or too high. For the sake of simplicity, we assume that $Y_0$, $\underline Y$, $\overline{Y}$ are multiple of $\xi$, and introduce the set $Q := \{ \underline Y, \underline Y + \xi, \dots, \overline Y - \xi, \overline Y \} \subset \mathbb{R}$.\\

We introduce $(L_t)_{t\in\mfT}$ the Doléans-Dade exponential 
$$L_t = \exp \left\{ \sum_{i=b,a} \int_0^t \log\left( \frac{\lambda^i_u}{a_0} \right) \d \hat N^i_u - \left(\lambda^i_u - a_0 \right) \d u\right\}\, ,$$
where
\begin{align}
    &\lambdaa_t = \barlambdaa(Y_{t\shortminus}, S_t)\,,\qquad \lambdab_t = \barlambdab(Y_{t\shortminus}, S_t)\,,\\
    &\qquad\quad \d \tilde{N}^{i}_t = \d \hat{N}^{i}_t - \lambda^{i}_t\, \d t\,, \quad \text{for } i=b,a,
\end{align}
with
\begin{align}\label{eq:intensity a}
    \barlambdaa(y, S) &= \max\left\{a_0, a_1 + a_2\,(S - c/y^2)\right\}\,,\\
    \label{eq:intensity b}
    \barlambdab(y, S) &= \max\left\{a_0, a_1 + a_2\,(c/y^2 - S)\right\}\,,\\
    &\qquad a_0>0, \,\,\, a_1, a_2\geq 0\,.
\end{align}
The function $c/y^2$ captures the marginal price in the pool because for $t\in\mfT$, we have that $Z_t = c/Y_t$. Thus, all else being equal, if  the price difference between the external price $S$ and the pool's marginal price $c/y^2$ is positive (resp.~negative), we expect that the intensity of trades that deposit (resp.~withdraw) asset $X$ in order to withdraw (resp.~deposit) asset $Y$ increases; conversely for the other direction.\footnote{ Modelling the stochastic intensities of trades arriving in the AMM has been done in  \cite{cartea2024strategic,aqsha2025equilibrium,baggiani2025optimal}. } \\

We then define the probability measure $\mathbb P$ given by
\[
\frac{\mathrm d\mathbb P}{\mathrm d\hat{\mathbb P}} := L_T\,.
\]
It follows that under $\mathbb P,$ the processes $\hatNb$ and $\hatNa$ have respective intensities $\left(\lambdab_t \right)_{t \in \mfT}$ and $\left(\lambdaa_t \right)_{t \in \mfT}$,  and therefore, the processes $N^b$ and $N^a$ have respective intensities $(\mathds 1_{\{Y_{t\shortminus} +\xi \le \overline y\}}\lambdab_t )_{t \in \mfT}$ and $( \mathds 1_{\{Y_{t\shortminus} -\xi \ge \underline y\}}\lambdaa_t )_{t \in \mfT}$, and $W$ is a Brownian motion. In the remaining of the paper, we only work under the probability measure $\mathbb P$.\\

\subsection{Problem formulation}

The impermanent loss (IL) at time $t\in \mfT$ is defined as
\begin{align}
 \text{IL}_t:=   -\big[X_t + Y_t\,S_t - \left(X_0 + Y_0\,S_t\right) \big] = -\Big[ \underbrace{X_t - X_0}_{P^X_t} + S_t \,\underbrace{\left( Y_t - Y_0 \right)}_{P^Y_t} \Big]\,.
\end{align}

To simplify the setup,\footnote{This approach has been used before, see e.g. \cite{fukasawa2025liquidity}, \cite{aqsha2025equilibrium}.} we consider the case of a representative agent who owns the entire liquidity available in the AMM and wants to close her position before time $T$. This representative liquidity provider (LP) wants to maximise her fee revenues while minimizing her IL.\\

Let $\Tc_{t,T}$ be set of stopping times taking values in $[t,T]$, and let $\Tc:=\Tc_{0,T}$. We are interested in solving the optimal stopping problem
\begin{equation}\label{pb_risk_neutral} 
    \sup_{\tau \in \Tc } \mathbb{E}\Big[ \underbrace{P^X_\tau +S_\tau P^Y_\tau}_{-\text{IL}_\tau} + R_\tau\Big]\,,
\end{equation}
where
$$R_t= \int_0^t \fee(Y_{s-}) \left(\d N^b_s + \d N^a_s \right)$$
with $\fee : \mathbb R_+ \rightarrow \mathbb R_+ $ a fee function representing the fees paid by liquidity takers (LTs) for each trade. In what follows we assume that $\fee$ has linear growth.\footnote{See \cite{baggiani2025optimal} for a study on optimal dynamic fees. One of their findings is that linear functions $\fee(y)$ are excellent approximations to the optimal fee structure. }\\

Notice that we only consider here the case of a risk-neutral liquidity provider. However, the case of a risk-averse liquidity provider can be treated in a very similar way and we provide the corresponding analysis in \ref{sec:app_risk_averse}, for a liquidity provider with an exponential utility function.

\section{Mathematical analysis}\label{sec: mathematical analysis}

In this section we characterise the value function of the optimal stopping problem in \eqref{pb_risk_neutral} as the unique viscosity solution to a Hamilton-Jacobi-Bellman Quasi-Variational Inequality (HJB QVI).

\subsection{Value function}

We write $-\text{IL}_t+R_t$ as
\begin{align*}
    P_t^X + S_t P_t^Y + R_t &= \displaystyle\int_0^t \left( \varphi_c(Y_{s^-} + \xi) - \varphi_c(Y_{s^-}) + \xi S_s + 
    \fee(Y_{s^-}) \right) \d N_s^b \\
    &\quad + \displaystyle\int_0^t \left( \varphi_c(Y_{s^-} - \xi) - \varphi_c(Y_{s^-}) - \xi S_s + \fee(Y_{s^-}) \right) \d N_s^a . \\
\end{align*}

We define the value function $v$ associated with \eqref{pb_risk_neutral}, as

\begin{align}\label{val_func_risk_neutral}
\begin{aligned}
v :\quad & \mfT \times Q \times \mathbb{R} \longrightarrow \mathbb{R} \\
& (t, y, S) \longmapsto \sup_{\tau \in \Tc_{t,T}} \mathbb{E} \left[ \int_{t}^{\tau} \left\{ 
\beta^b\left(Y^{t,y}_{s^-}, S^{t,S}_s\right) \, \lambda^b_s \, \mathds{1}_{\{ Y^{t,y}_{s\shortminus} + \xi \le \overline{Y} \}} 
+ \beta^a\left(Y^{t,y}_{s^-}, S^{t,S}_s\right) \, \lambda^a_s \, \mathds{1}_{\{ Y^{t,y}_{s\shortminus} - \xi \ge \underline{Y} \}} 
\right\} \mathrm{d}s \right],
\end{aligned}
\end{align}
where 
\begin{equation}\label{eq:betas}
    \begin{cases}
        \beta^b(y,S) := \varphi_c(y+ \xi) - \varphi_c(y) + \xi S + \fee(y),\\
        \beta^a(y, S) := \varphi_c(y - \xi) - \varphi_c(y) - \xi S + \fee(y).      
    \end{cases}
\end{equation}

\begin{definition}
    We denote by $\Xi$ the set of functions $u:\mfT\times Q \times \mathbb R \rightarrow \mathbb R$ that are non-negative, bounded in their  second variable, and with at most quadratic growth in their third variable, i.e. there exists $C_0$ such that
    $$0\le u(t,y,S) \le C_0 (1 + S^2)$$
    for all $(t,y,S) \in \mfT\times Q \times \mathbb R.$

\end{definition}

\begin{proposition}\label{v_bounded}
The value function \( v : \mfT \times Q \times \mathbb{R} \to \mathbb{R} \) is in $\Xi$. In particular, it is well-defined.
\end{proposition}

\begin{proof}
Let 
$(t,y,S) \in \mfT \times Q \times \mathbb{R}$. The value function is defined as
\begin{align}
    v(t,y,S) = \sup_{\tau \in \Tc_{t,T}} \mathbb{E} \left[ \displaystyle\int_{t}^{\tau} \left\{ \beta^b\left(Y^{t,y}_{u\shortminus}, S^{t,S}_u\right) \, \lambda^b_u \mathds 1_{\{ Y^{t,y}_{u\shortminus} +\xi \le \overline{Y} \}} + \beta^a\left(Y^{t,y}_{u\shortminus}, S^{t,S}_u\right) \, \lambda^a_u  \mathds 1_{\{ Y^{t,y}_{u\shortminus} - \xi \ge \underline{Y} \}} \right\} \d u \right].
\end{align}

Since we assumed that $\fee$ has linear growth and that $(Y_t)_{t \in \mathfrak{T}}$ is bounded, it follows that the value function is bounded in its second variable. Moreover, we can show that both terms inside the integrand exhibit at most quadratic growth with respect to the third variable. Observe that 
\begin{align}
& \beta^a\left(Y^{t,y}_{u\shortminus}, S^{t,S}_u\right) \, \bar{\lambda}^a\left(Y^{t,y}_{u\shortminus}, S^{t,S}_u\right) \, \mathds 1_{\{ Y_{u\shortminus} - \xi \ge \underline{Y} \}} \nonumber \\ 
& \,\,\,\, = \,\left[ \varphi_c\left(Y^{t,y}_{u\shortminus} - \xi\right) - \varphi_c\left(Y^{t,y}_{u\shortminus}\right) - \xi S^{t,S}_u + \fee\left(Y^{t,y}_{u\shortminus}\right) \right] \, \max\left\{a_0, a_1 + a_2\,\left(S^{t,S}_u - \frac{c}{{\left(Y^{t,y}_{u{\shortminus}}\right)}^2}\right)\right\}\, \mathds 1_{\{ Y_{u\shortminus} - \xi \ge \underline{Y} \}} \, \nonumber \\
& \,\,\,\, \leq \, \left[ \frac{c}{\underline{Y}} + \frac{c}{\underline{Y} + \xi} + \xi\left|S^{t,S}_u\right| + \left|\underset{y \in Q}{\max} \, \fee\left(y\right)\right| \right]  \, \left[a_0 + a_1 + a_2\,\left(\left|S^{t,S}_u\right| + \frac{c}{(\underline{Y} + \xi)^2}\right) \right] \, \nonumber \\
& \,\,\,\, \leq \, C^a_0 + C^a_1\sup_{u \in {[t,\tau]}}\left|S^{t,S}_u\right| + C^a_2\sup_{u \in {[t,\tau]}}\left|S^{t,S}_u\right|^2, \nonumber
\end{align}
where $C^a_0$, $C^a_1$, $C^a_2$ are constant depending on $\overline Y$, $\underline Y$, $\xi$, $c$, $a_0$, $a_1$ and $a_2$. Similarly, 
\begin{align}
& \beta^b(Y^{t,y}_{u\shortminus}, S^{t,S}_u) \, \bar{\lambda}^b(Y^{t,y}_{u\shortminus}, S^{t,S}_u) \, \mathds 1_{\{ Y_{u\shortminus} + \xi \le \overline{Y} \}} \nonumber \\
& \,\,\,\, = \, \left[ \varphi_c\left(Y^{t,y}_{u\shortminus} + \xi\right) - \varphi_c\left(Y^{t,y}_{u\shortminus}\right) + \xi S^{t,S}_u + \fee(Y^{t,y}_{u{\shortminus}}) \right] \, \max\left\{a_0, a_1 + a_2\,\left(\frac{c}{{\left(Y^{t,y}_{u{\shortminus}}\right)}^2} - S^{t,S}_u\right)\right\}\ \, \mathds 1_{\{ Y_{u\shortminus} + \xi \le \overline{Y} \}} \, \nonumber \\
& \,\,\,\, \leq \, \left[ \left|\frac{c}{\underline{Y} + \xi}\right| + \left|\frac{c}{\underline{Y}}\right| + \xi \left|S^{t,S}_u\right| + \left|\underset{y \in Q}{\max} \, \fee(y)\right| \right] \, \left[a_0 + a_1 + a_2\,\left(\frac{c}{\underline{Y}^2} + \left|S^{t,S}_u\right|\right)\right] \, \nonumber \\
& \,\,\,\, \leq \, C^b_0 + C^b_1\sup_{u \in {[t,\tau]}}\left|S^{t,S}_u\right| + C^b_2\sup_{u \in {[t,\tau]}}\left|S^{t,S}_u\right|^2, \nonumber
\end{align}
where $C^b_0$, $C^b_1$, $C^b_2$ are constant depending on $\overline Y$, $\underline Y$, $\xi$, $c$, $a_0$, $a_1$ and $a_2$.\\

Using Jensen for concave functions and then Doob's inequality we have that
\begin{align}
\mathbb{E} \left[ \sup_{u \in {[t,T]}} \left|S^{t,S}_u\right| \right] &\leq \, \left|S_0\right| + \sigma\mathbb{E} \left[ \sup_{u \in {\mfT}} \left|W_u\right| \right] = \, \left|S_0\right| + \sigma\mathbb{E} \left[ \sup_{u \in {\mfT}} \sqrt{W_u^2} \right] \leq \, \left|S_0\right| + 2\sigma \sqrt{T}, \nonumber\\
\mathbb{E} \left[ \sup_{u \in {[t,T]}} \left|S^{t,S}_u\right|^2 \right] &\leq \, S_0^2 + 2 \left|S_0\right| \sigma\mathbb{E} \left[ \sup_{u \in {\mfT}} \left|W_u\right| \right] + \sigma^2 \mathbb{E} \left[ \sup_{u \in {\mfT}} W_u^2 \right] \leq \, S_0^2 + 4\left|S_0\right|\sigma \sqrt{T} + 4\sigma^2T. \nonumber
\end{align}

We  conclude that the value function has quadratic growth in its third variable and then is locally bounded from above.\\

Furthermore, it is clear in the definition of the value function in \eqref{val_func_risk_neutral} that $v(t,y,S)\geq0$ (it suffices to consider the stopping time $\tau = t$). Therefore, the value function is bounded from below by $0$.
\end{proof}

Following \cite{touzi2012optimal} (Theorem 4.3), we now state the dynamic programming principle for this optimal stopping problem.

\begin{lemma} \label{lemma_dpp}
     Let $(t,y,S) \in [0,T) \times Q \times \mathbb{R}$, then for all $\theta \in \mathcal{T}_{t,T}$ we have
    \begin{align}
        v(t, y, S) = & \, \sup_{\tau \in \Tc_{t,T}} \mathbb{E}\bigg[\int_{t}^{\tau \wedge \theta}  \left[ \mathds{1}_{\{Y_{u\shortminus} + \xi \geq \underline{Y}\}} \, \beta^b(Y^{t,y}_{u\shortminus},S^{t,S}_u) \lambda^b_u + \mathds{1}_{\{Y_{u\shortminus} - \xi \leq \overline{Y}\}} \, \beta^a(Y^{t,y}_{u\shortminus}, S^{t,S}_u)\lambda^a_u \right] \d u \bigg. \nonumber \\
        & \hspace{9cm} + \, \mathds{1}_{\{\tau \geq \theta\}}v(\theta, Y^{t, y}_{\theta}, S^{t, y}_{\theta}) \bigg. \bigg] \nonumber.
    \end{align}
\end{lemma}

In particular, the candidate HJB QVI for our problem is
\begin{align}
        0 = \min\bigg\{& -\frac{\partial}{\partial t} v(t, y, S) \, - \frac{1}{2} \sigma^2 \frac{\partial^2}{\partial S^2} v(t, y, S) \,- \mathds 1_{\{y +\xi \le \overline Y\}}\bar{\lambda}^b(y, S) \big[ \beta^b(y,S) + v(t, y +\xi, S) - v(t, y, S) \big] \, \nonumber \\
    & - \mathds 1_{\{y -\xi \ge \underline Y\}} \bar{\lambda}^a(y, S) \big[ \beta^a(y,S) + v(t, y - \xi, S) - v(t, y, S) \big]\,  , \, v(t,y,S) \bigg\} \qquad \text{on } [0,T) \times  Q \times \mathbb R \label{eq:hjb_qvi},
\end{align}
with terminal condition $v(T,y,S)=0$ for all $(y,S) \in Q \times \mathbb R$.

\subsection{Viscosity solution}

 In this section, we study the viscosity properties of the value function \eqref{val_func_risk_neutral}. We start by introducing a set of test functions and defining properly the notion of viscosity solution for the HJB QVI \eqref{eq:hjb_qvi}.

\begin{definition}
We denote by  $\mathcal{C}$ the class of functions $\gamma : (t,y,S) \in \mfT \times Q \times \mathbb{R} \mapsto \gamma (t,y,S) \in \mathbb{R}$
such that $\gamma$ is continuously differentiable with respect to the first variable on \([0,T)\), $\gamma$ is twice continuously differentiable with respect to the third variable on \(\mathbb{R}\) and $\left|\gamma\right|$, $\left|\frac{\partial}{\partial S}\gamma\right|$ are bounded.  
\end{definition}

\begin{definition}
    For a locally bounded function $u : \mfT \times Q \times \mathbb{R} \to \mathbb{R}$, we denote by $u_*$ and $u^*$ the lower and upper semicontinuous envelopes of $u$, defined as follows:
        \begin{equation}
            u_*(t,y,S) = \liminf_{(t', S') \to (t, S)} u(t', y, S'), \qquad u^*(t,y,S) = \limsup_{(t', S') \to (t, S)} u(t', y, S')
        \end{equation}
        for all $(t,y,S) \in \mfT\times Q \times \mathbb R$.
\end{definition}

\begin{definition}
\begin{enumerate}[label=(\roman*)]
\item Let $u$ be an upper semicontinuous (USC) function on~ $\mfT \times Q \times \mathbb{R}$. We say that $u$ is a viscosity subsolution to the HJB QVI on~ $[0,T) \times Q \times \mathbb{R}$ if for all $(\tilde{t},\tilde{y},\tilde{S}) \in [0,T) \times Q \times \mathbb{R}$, and for all $\gamma \in \mathcal{C}$ such that ~$(u-\gamma)(\tilde{t},\tilde{y},\tilde{S})=\max_{(t,y,S) \in [0,T) \times Q \times \mathbb{R}}(u-\gamma)(t,y,S)=0$, we have
\begin{align}
\min \Bigg\{ 
&-\frac{\partial}{\partial t} \gamma(\tilde{t}, \tilde{y}, \tilde{S}) - \frac{1}{2} \sigma^2 \frac{\partial^2}{\partial S^2} \gamma(\tilde{t}, \tilde{y}, \tilde{S}) \nonumber \, \\
&  - \mathds 1_{\{y +\xi \le \overline Y\}}\bar{\lambda}^b(\tilde{y}, \tilde{S}) \left[ \beta^b(\tilde{y},\tilde{S}) + \gamma(\tilde{t}, \tilde{y} + \xi, \tilde{S}) - \gamma(\tilde{t}, \tilde{y}, \tilde{S}) \right]\\ 
& - \mathds 1_{\{y -\xi \ge \underline Y\}}\bar{\lambda}^a(\tilde{y}, \tilde{S}) \left[ \beta^a(\tilde{y},\tilde{S}) + \gamma(\tilde{t}, \tilde{y} - \xi, \tilde{S}) - \gamma(\tilde{t}, \tilde{y}, \tilde{S}) \right], \, \gamma(\tilde{t}, \tilde{y}, \tilde{S}) \Bigg\} \leq 0. \nonumber
\end{align}
\item Let $w$ be a lower semicontinuous (LSC) function on $\mfT \times Q \times \mathbb{R}$. We say that $w$ is a viscosity supersolution to the HJB QVI on $[0,T) \times Q \times \mathbb{R}$ if for all $(\tilde{t},\tilde{y},\tilde{S}) \in \mfT \times Q \times \mathbb{R}$, and for all $\gamma \in \mathcal{C}$ such that $(w-\gamma)(\tilde{t},\tilde{y},\tilde{S})=\min_{(t,y,S) \in [0,T) \times Q \times \mathbb{R}}(w-\gamma)(t,y,S)=0$, we have
\begin{align}
\min \Bigg\{ 
&-\frac{\partial}{\partial t} \gamma(\tilde{t}, \tilde{y}, \tilde{S}) - \frac{1}{2} \sigma^2 \frac{\partial^2}{\partial S^2} \gamma(\tilde{t}, \tilde{y}, \tilde{S}) \nonumber \, \\
&  - \mathds 1_{\{y +\xi \le \overline Y\}}\bar{\lambda}^b(\tilde{y}, \tilde{S}) \left[ \beta^b(\tilde{y},\tilde{S}) + \gamma(\tilde{t}, \tilde{y} + \xi, \tilde{S}) - \gamma(\tilde{t}, \tilde{y}, \tilde{S}) \right]\\ 
& - \mathds 1_{\{y -\xi \ge \underline Y\}}\bar{\lambda}^a(\tilde{y}, \tilde{S}) \left[ \beta^a(\tilde{y},\tilde{S}) + \gamma(\tilde{t}, \tilde{y} - \xi, \tilde{S}) - \gamma(\tilde{t}, \tilde{y}, \tilde{S}) \right], \, \gamma(\tilde{t}, \tilde{y}, \tilde{S}) \Bigg\} \geq 0. \nonumber
\end{align}

\item If $v$ is locally bounded on $\mfT \times Q \times \mathbb{R}$, we say that $v$ is a viscosity solution to HJB QVI on $[0,T) \times Q \times \mathbb{R}$ if its upper semicontinuous envelope $v^*$ and its lower semicontinuous envelope $v_*$ are respectively subsolution on $[0,T) \times Q \times \mathbb{R}$ and supersolution on $[0,T) \times Q \times \mathbb{R}$ to the HJB QVI.
\end{enumerate}
\end{definition}

We state below the main result of this section. Since the proof is long and technical, we include it in \ref{sec:proof_viscosity}.

\begin{theorem}\label{thm_viscosity}
    The value function $v$ defined in  \eqref{val_func_risk_neutral} is the only viscosity solution in $\Xi$ to the HJB QVI \eqref{eq:hjb_qvi} with terminal condition $v(T,y,S) = 0$ for all $(y,S) \in Q \times \mathbb R$. Moreover, it is continuous.
\end{theorem}

\section{Numerical results}\label{sec: numerical results}

\subsection{Euler scheme and Longstaff-Schwartz algorithm} 

We design an implicit Euler scheme on a three-dimensional grid -- time $t$, reserve $Y$, and price $S$ -- for the HJB QVI and apply Neumann conditions at the boundaries in $S$. The HJB~QVI~$\eqref{eq:hjb_qvi}$ features a standard diffusion term coupled with a jump component. To handle this structure, we employ an operator splitting method at each time step: the jump terms are treated explicitly, followed by an implicit step for the diffusion part. Finally, the minimum condition is enforced to determine whether it is optimal to stop or continue.\\

The Euler scheme may become unstable for certain parameter values due to the Courant-Friedrichs-Lewy condition. To address this instability and enable the analysis of how different parameters influence the model's output, we also implement the Longstaff--Schwartz algorithm. This approach allows us to  approximate the solution numerically by leveraging the dynamic programming principle presented in Lemma~\ref{lemma_dpp}. It leads to an iterative decision process in which, at each discretised time step, a choice must be made between continuing or stopping. This process is captured by the following recursive scheme

\begin{align}
    \hat{v}(t,y,S) = & \max \Bigg\{ \mathbb{E} \left[ \int_{t}^{t+\delta} \left\{ 
    \beta^b\left(Y^{t,y}_{s^-}, S^{t,S}_s\right) \, \lambda^b_s \, \mathds{1}_{\{ Y^{t,y}_{s^-} + \xi \le \overline{Y} \}} + \beta^a\left(Y^{t,y}_{s^-}, S^{t,S}_s\right) \, \lambda^a_s \, \mathds{1}_{\{ Y^{t,y}_{s^-} - \xi \ge \underline{Y} \}} \right\} \, \mathrm{d}s \right. \nonumber \\
    & \hspace{26em} \left.   + \, \hat{v}\left(t+\delta, Y_{t+\delta}^{t, y}, S_{t+\delta}^{t, S}\right) \right], 0 \Bigg\}, \nonumber
\end{align}

with the terminal condition $\hat{v}(T,y,S) = 0$, since no further fees are earned and no additional impermanent loss is incurred when exiting at $T$. Here $\delta>0$ denotes the time step of the scheme.\\

The expression above involves a conditional expectation that must be evaluated at each time step in order to make an optimal decision. This expectation can be approximated using polynomial regression within a Monte Carlo simulation framework. The approach was originally introduced in~\cite{longstaffValuingAmericanOptions1998} for pricing American options. Algorithm~\ref{algo_ls_amm} outlines the adapted pseudo-code based on this methodology, tailored to our optimal exit time problem.\\

\begin{algorithm}
\caption{Longstaff-Schwartz Algorithm}
\begin{algorithmic}[1]
\Require $n$, $m$, $T$, $d$, $\sigma$, $a_0$, $a_1$, $a_3$, $\xi$, $X_0$, $Y_0$
\State Initialise matrix $\mathbf{A}$, $\mathbf{S}$, $\mathbf{Y}$, $\mathbf{V}$, $\boldsymbol{\tau}$ with zeros and shape $(m,n)$, to store trajectories of the performance $-\mathrm{IL}+R$, oracle price $S$, reserve $Y$, value function $v$, and stopping time $\tau$, respectively.
\State Set the time step size at $T/n$
\State Generate $m$ trajectories for $\mathbf{A}$, $\mathbf{S}$, $\mathbf{Y}$\Comment{Matrix $\mathbf{A}$ stores the trajectories of $\shortminus \mathrm{IL}+R$}
\State Set $\mathbf{V}_n \gets 0$ \Comment{Terminal condition}
\State Set $\boldsymbol{\tau}_n \gets 1$ \Comment{Exit at $T$ if not done before}

\For{$i \gets n-1$ to $1$} for all trajectories
    \State Generate a polynomial $\mathbf{P}_i$ of degree $d$ in $\mathbf{S}_i$ and $\mathbf{Y}_i$
    \State Perform polynomial regression of $\mathbf{V}_{i+1}$ on $\mathbf{P}_i$
    \State Estimate continuation value $\mathbf{C}_i$ \Comment{Approximation of $\mathbb{E}[\mathbf{V}_{i+1}\big| S_i, Y_i]$}
    \State Get trajectories where $\mathbf{C}_i \le 0$ \Comment{Stopping rule}
    \State Update value function:
        \[
        \mathbf{V}_i \gets
        \begin{cases}
        0 & \text{if the trajectory has stopped,} \\
        \mathbf{A}_{i+1} - \mathbf{A}_{i} + \mathbf{V}_{i+1} & \text{otherwise.}
        \end{cases}
        \]
    \State Update stopping time: $\boldsymbol{\tau}_i \gets 1$ if the trajectory has stopped
\EndFor
\State \Return $\mathbf{V}$, $\boldsymbol{\tau}$
\end{algorithmic}
\end{algorithm} \label{algo_ls_amm}

To test the Longstaff–Schwartz algorithm against a well-known Euler scheme, we employ the following model parameters: $n = 1{,}440$, $m = 5{,}000$, $T = 1$ day, $d = 3$, $\sigma = 100\; \$ \cdot \text{day}^{-1/2}$, $a_0 = 4\;\text{day}^{-1}$, $a_1 = 8\;\text{day}^{-1}$, $a_2 = 0.04\;\$^{-1}\cdot\text{day}^{-1}$, $\xi = 1$, $X_0 = 1{,}000\; \$$, and $Y_0 = 1{,}000$ coins. Here we consider a fictitious representative cryptocurrency; in the numerical experiments below we deploy the strategies with calibrated parameters.\\

The right panel in Figure~\ref{fig:3dplot-euler-vs-ls} shows that both methods produce similar overall shapes and identical exit regions. The red curves coincide with different slices of the Euler surface. The black curve shows that the value function is maximised when $S=c/Y^2(=Z)$, that is when the AMM price is equal to the external price. The LP leaves the pool when the value function reaches 0, that is when $|S-Z|$ becomes too large. This can be interpreted in the following way: in our model,  price formation happens on the external venue, thus, when this price deviates from the internal price, this creates a ``potential'' impermanent loss that is only realised once arbitrageurs  trade on the AMM to align the internal price with the external price. Therefore, when large price move happen, LPs might be driven to leave the pool before the IL is realised, that is, before arbitrageurs can act on the AMM. Lastly, on the right panel, we observe that the value function decreases with time and the hold region shrinks as we get close to time $T$. This occurs because as time progresses there is less expected fees to be collected, reducing the incentives for the LP to tolerate high levels of potential impermanent loss.

\begin{figure}[H]
    \centering
    \includegraphics[width=0.49\linewidth]{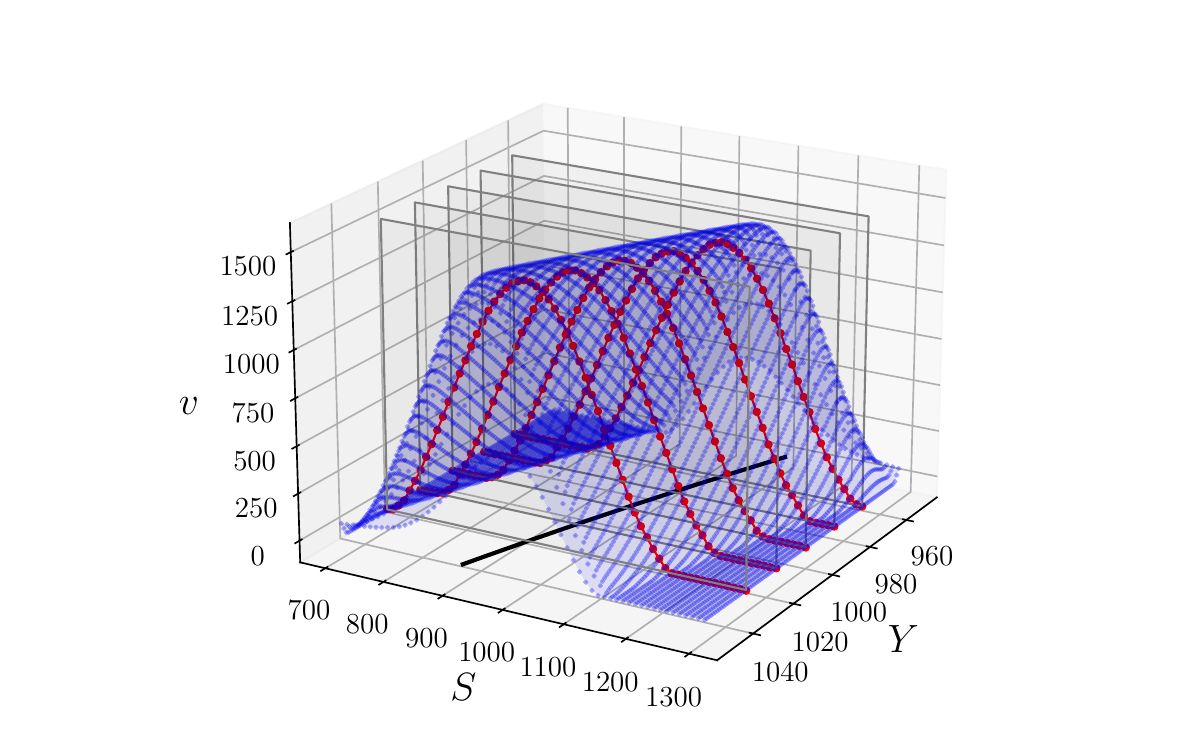}
    \includegraphics[width=0.49\linewidth]{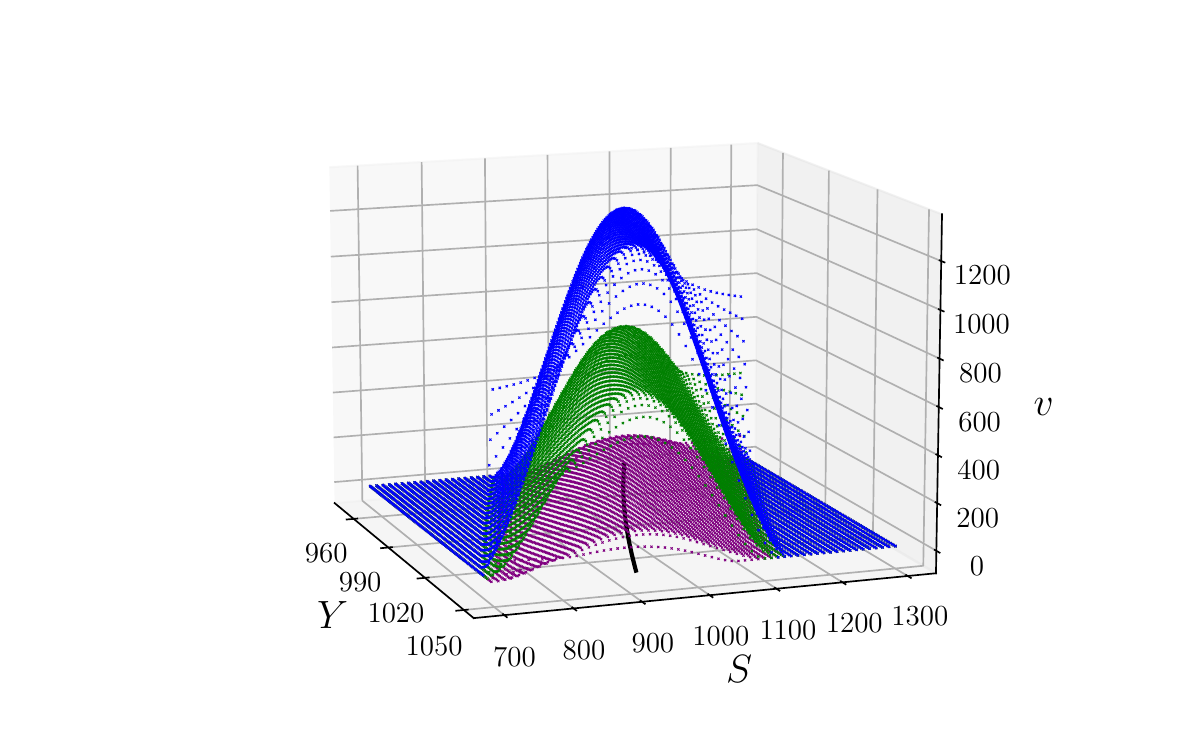}
    \caption{Left shows the value function $v$ at $t = 0$ computed with the Euler method (blue surface) and Longstaff--Schwartz algorithm (red slices) with the black curve in the $S\text{-}Y$ plane representing states where $S = c / Y^2$, while right shows the Euler method value function at $t = 0$ (blue), $t = 0.5$ (green), and $t = 0.9$ (purple).}
    \label{fig:3dplot-euler-vs-ls}
\end{figure}

Figure~\ref{fig:2d-plots-euler-vs-ls}  confirms the previous observation. We also see that the Longstaff--Schwartz method underestimates the value function when compared to the Euler method. This discrepancy can be explained by the discretization of time: the liquidity provider (LP) can only exit the pool at discrete time steps, rather than continuously. However, the stopping regions defined by the two methods are  aligned.\\

\begin{figure}[H]
    \centering
    \includegraphics[width=0.32\linewidth]{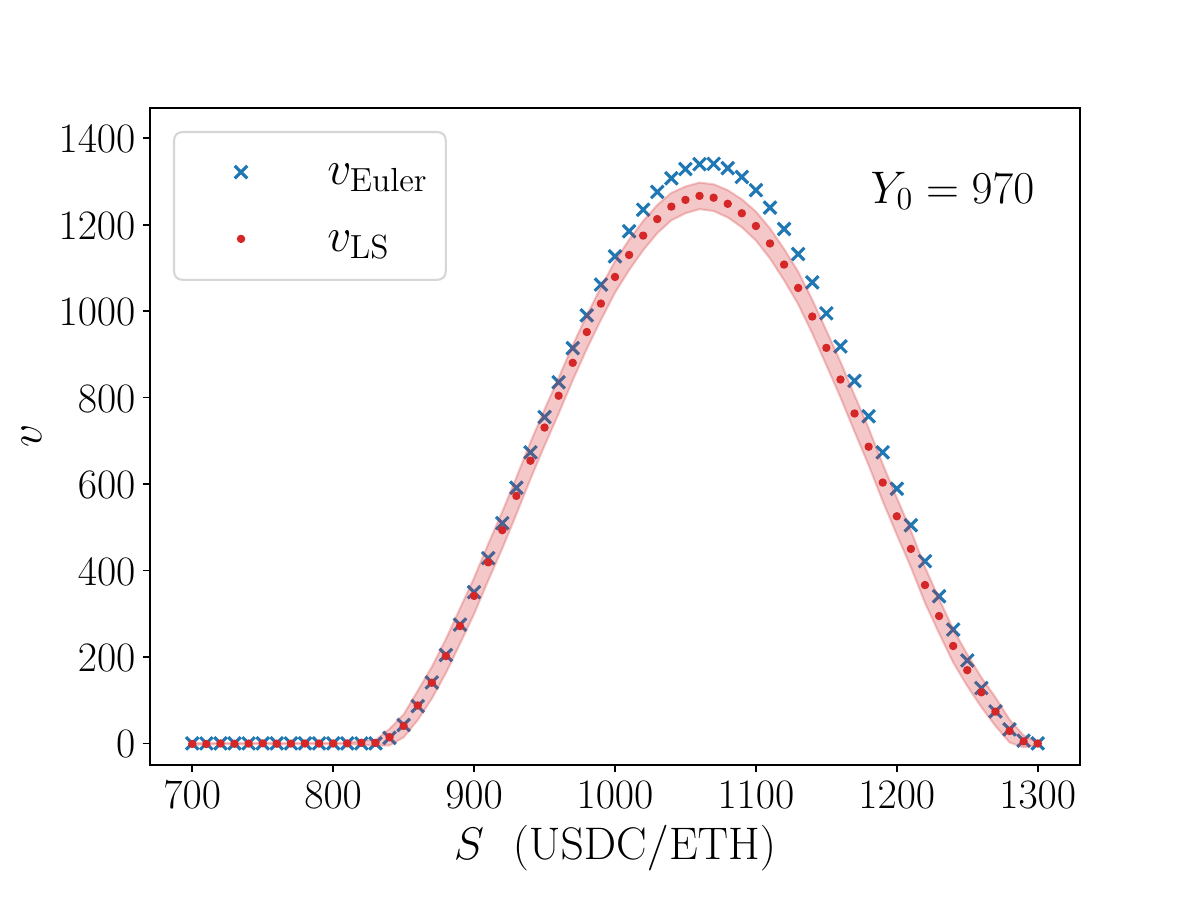}
    \includegraphics[width=0.32\linewidth]{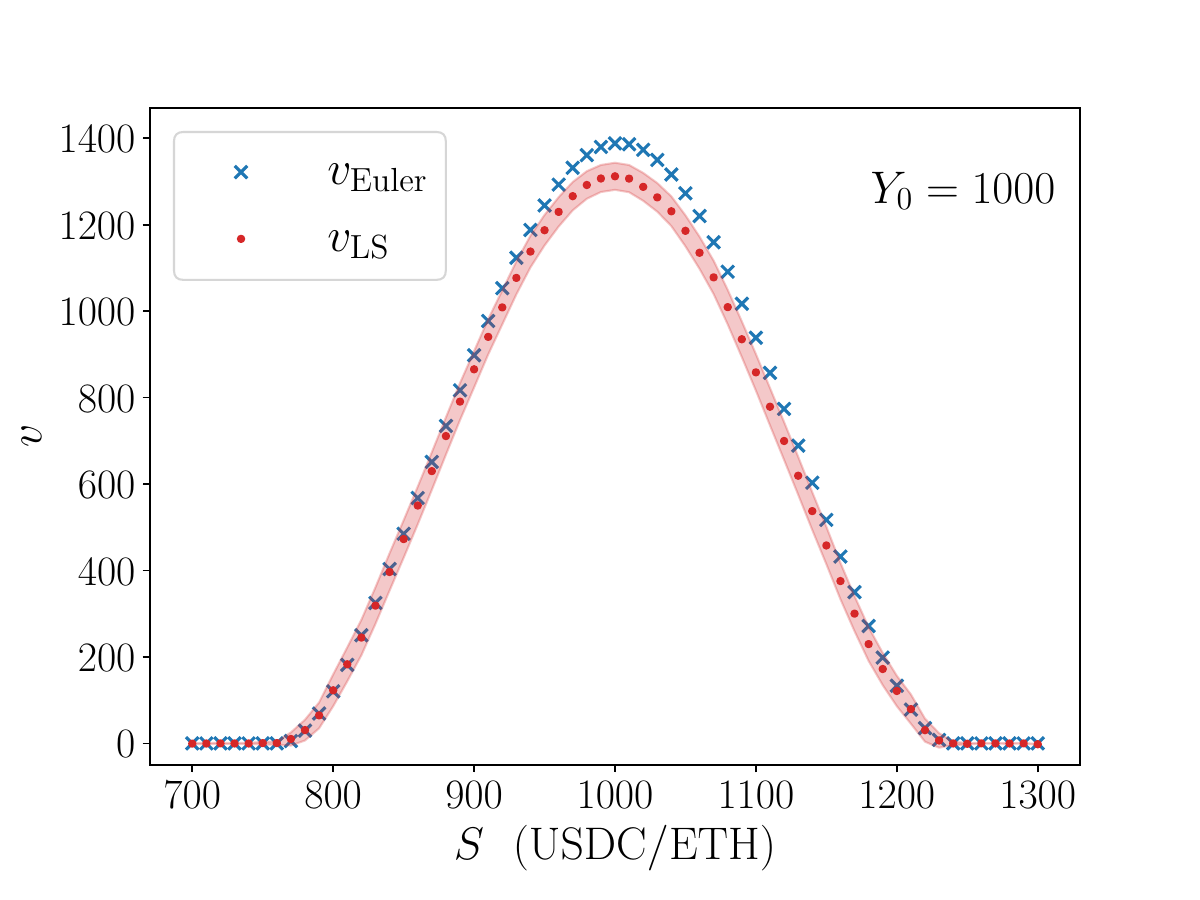}
    \includegraphics[width=0.32\linewidth]{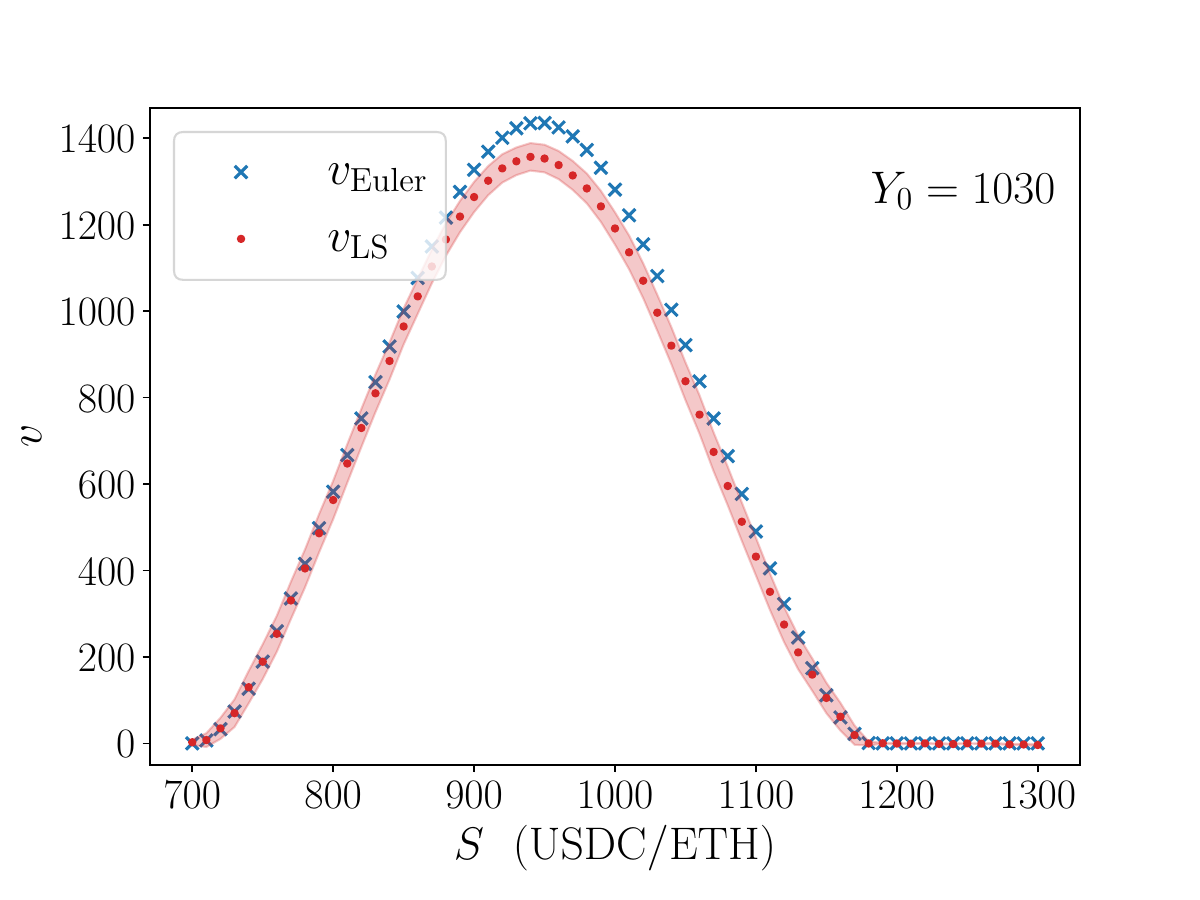}
    \caption{Comparison of the Euler method (blue crosses) and the Longstaff--Schwartz method (red dots) at time $t=0$ for three initial reserve values $Y$ (displayed in the legend). Each plot includes a 5\% confidence interval for the Longstaff--Schwartz method.}
    \label{fig:2d-plots-euler-vs-ls}
\end{figure}

\subsection{Comparative statics}

For the experiments below, we discretise $\mfT$ in 1,440 timesteps with $T = 1$ (one day).  All confidence intervals and averages are computed using 10,000 simulations. We use the model parameters from \cite{aqsha2025equilibrium} who used market data from Binance and Uniswap V2 in the pair ETH-USDC between 1 January 2022 and 30 April 2022 to calibrate model parameters, that is, we take $S_0=Z_0=2820\; \$$, $\sigma = 0.0569\,S_0\; \$ \cdot \text{day}^{-1/2}$, $Y_0=50,000\;\text{ETH}$, and $X_0= Y_0\,Z_0\; \$$. Below, the baseline values for the parameters in the intensities are: $a_0 = 1\;\text{day}^{-1}$, $a_1= 10\;\text{day}^{-1}$, and $a_2=10\;\$^{-1}\cdot\text{day}^{-1}$; below, we carry out robustness checks where we stress the key model parameters $a_1$ (modulating noise traders) and $a_2$ (modulating arbitrageurs). Lastly, we take the fee to be constant and equal to $0.01 \times\xi\times S_0$ with $\xi=100$.\\

Figure \ref{fig:sample-paths} shows a sample path (together with 5\% and 95\% quantile bands) for the key processes involved. As expected, we see that the marginal price in the pool $Z_t$ follows closely the external price $S_t$. This is due to the arbitrageurs that align quotes and are modulated through the value of the model parameter $a_2$. The middle panel shows that increase in prices are accompanied by depletion of reserves and vice-versa. The right panel shows the fee collected is increasing which is the main mechanism to offset the impermanent loss also shown in that panel. 

\begin{figure}[H]
    \centering
    \includegraphics[width=0.32\linewidth]{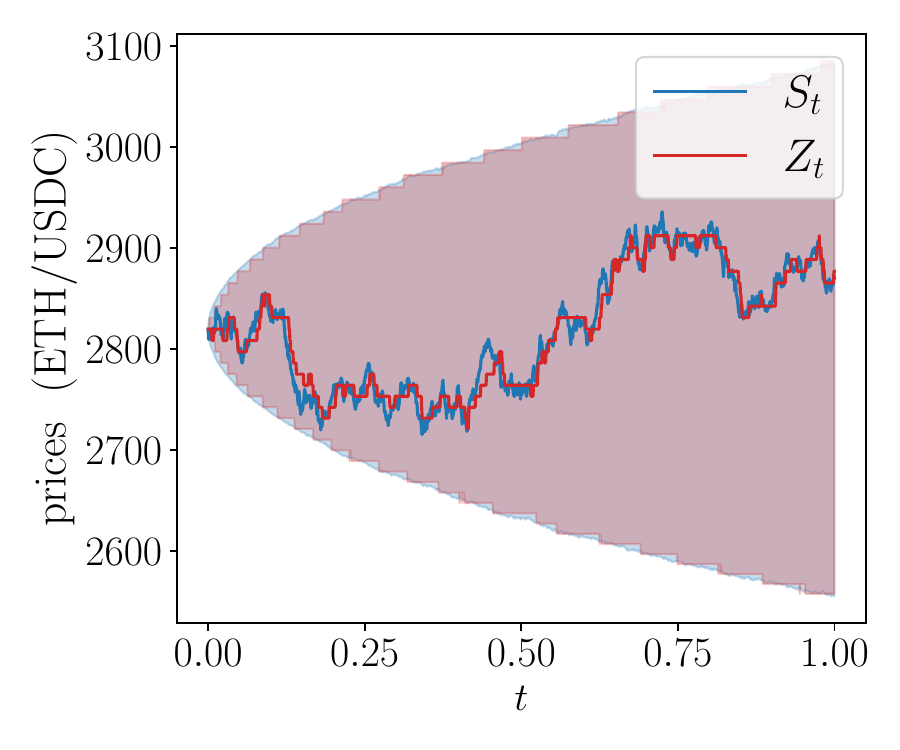}
    \includegraphics[width=0.32\linewidth]{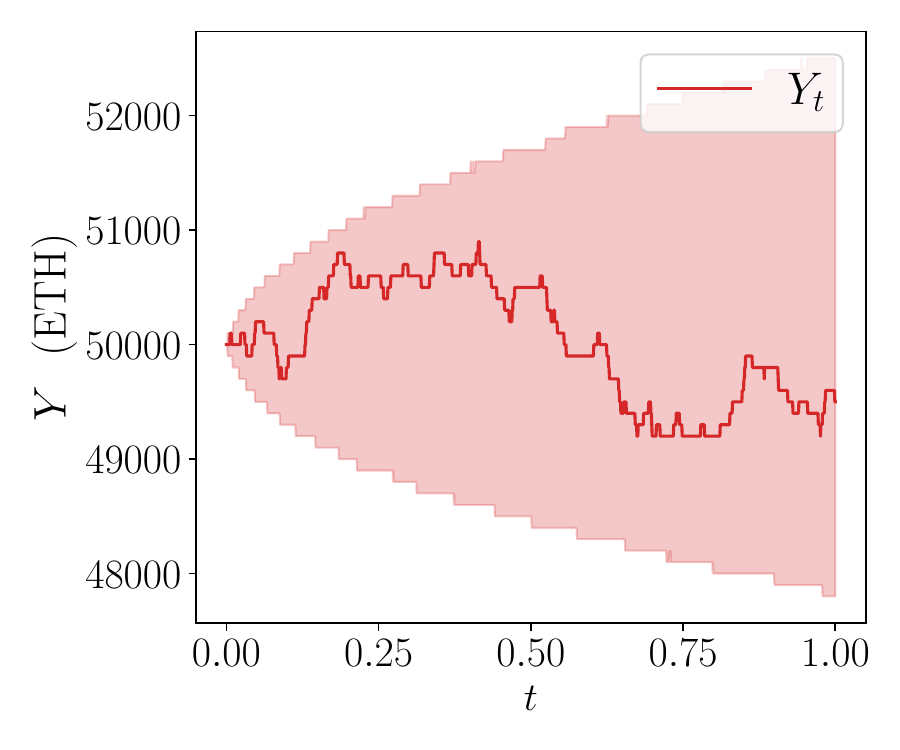}
    \includegraphics[width=0.32\linewidth]{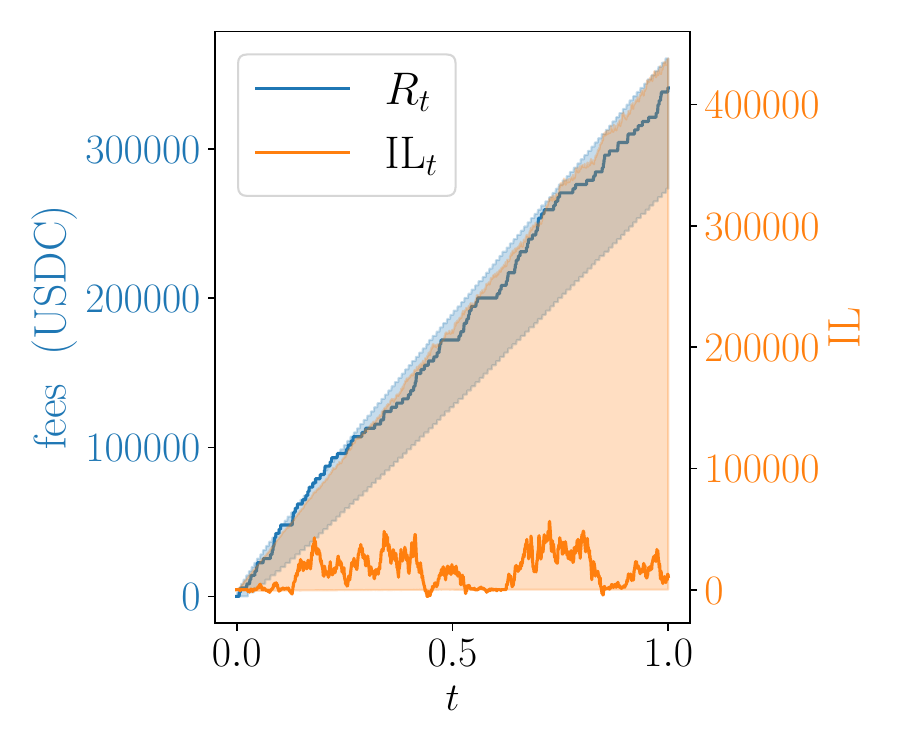}
    \caption{Sample paths for the price processes $S_t$ and $Z_t$, the inventory process $Y_t$ (ETH), and the fees collected by the LP. All processes are accompanied by the running quantiles (5\% and 95\%) across time.}
    \label{fig:sample-paths}
\end{figure}

Next, Table \ref{tab:sigma-table} explores the effects of the volatility of the oracle price in the expected exit time $\mathbb{E}[\tau]$ (here $\tau=T=1$ if the LP does not exit the pool before $T$), the expected total collected fees $\mathbb{E}[R_\tau]$, and the expected impermanent loss $\mathbb{E}[\mathrm{IL}_\tau]$. 

\begin{table}[H]
    \centering
    \begin{tabular}{l|rrr}
\hline
\hline
$\sigma_{\mathrm{new}}$ & $\mathbb{E}[\tau]$ & $\mathbb{E}[R_\tau]$ & $\mathbb{E}[\mathrm{IL}_\tau]$ \\
\hline
\hline
$\sigma$/5 & 1.00 (0.00) & 161,765 (22,320) & 4,364 (6,419) \\
$\sigma$/4 & 1.00 (0.00) & 168,343 (22,217) & 6,912 (10,024) \\
$\sigma$/3 & 1.00 (0.00) & 181,148 (21,806) & 12,422 (17,830) \\
$\sigma$/2 & 1.00 (0.00) & 210,406 (22,535) & 28,177 (40,154) \\
$\sigma$ & 1.00 (0.01) & 316,222 (26,930) & 112,635 (157,844) \\
$2\,\sigma$ & 0.51 (0.27) & 273,003 (144,005) & 202,999 (372,645) \\
$3\,\sigma$ & 0.01 (0.01) & 3,599 (9,363) & 2,449 (13,042) \\
$4\,\sigma$ & 0.00 (0.00) & 1,794 (3,984) & 1,304 (7,411) \\
$5\,\sigma$ & 0.00 (0.00) & 532 (1,198) & 420 (1,401) \\
\hline
\end{tabular}
    \caption{Summary statistics for the expected (i) exit time, (ii) fees collected, and (iii) impermanent loss, as $\sigma$ varies. Mean values (with standard deviation in parenthesis) across 10,000 simulations.  }
    \label{tab:sigma-table}
\end{table}

There is a concave relationship between $\sigma$ and (i) the expected collected fees $\mathbb{E}[R_\tau]$ together with the (ii) impermanent loss $\mathbb{E}[\mathrm{IL}_\tau]$. Indeed, for the first half of the table the higher sigma the more fees there are collected because of the fees paid by arbitrageurs. However,  as $\sigma$ becomes large, this effect disappears because the LPs exits the pool early in all sample paths due to the sharp increase of the impermanent loss (see right hand side column). \\

Next, we investigate the role of the fees in the optimal exit time and the profitability of LPs. All else being equal, as the value of the model parameter $\fee$ increases, we expect the overall order flow in the market to decrease. In our model, this implies that the higher the value of $\fee$, the lower the values of $a_0$, $a_1$, and $a_2$. To convey this stylised fact, when we change the value of $\fee$ to a given $\tilde{\fee}$, we change the value of $a_i$ to be
$$
a_i \,\exp(-\kappa(\tilde\fee - \fee))\,,\qquad i\in\{0,1,2\}\,,
$$
for $\kappa = 10^{-4}\, \$^{-1}$. We take this value because in our experiments the unstressed value of $\fee$ is in the order of $10^{4}$, but the qualitative behaviour we observe is robust to this choice. Similarly, the exponential form is again for simplicity; we observe a similar pattern when employing a linear decay schedule.
In Table \ref{tab:fee-table}, we explore the effect of the fees charged by LPs.

\begin{table}[H]
    \centering
    \begin{tabular}{l|rrr}
\hline
\hline
$\fee_{\mathrm{new}}$ & $\mathbb{E}[\tau]$ & $\mathbb{E}[R_\tau]$ & $\mathbb{E}[\mathrm{IL}_\tau]$ \\
\hline
\hline
$\fee$/5 & 0.01 (0.01) & 373 (750) & 173 (1,308) \\
$\fee$/4 & 0.04 (0.05) & 3,027 (5,099) & 2,239 (11,377) \\
$\fee$/3 & 0.11 (0.12) & 12,578 (14,578) & 8,437 (35,736) \\
$\fee$/2 & 0.96 (0.09) & 166,560 (20,956) & 105,643 (141,534) \\
$\fee$ & 1.00 (0.01) & 316,222 (26,930) & 112,635 (157,844) \\
$2\,\fee$ & 1.00 (0.00) & 530,160 (50,735) & 113,269 (161,299) \\
$3\,\fee$ & 1.00 (0.00) & 670,137 (71,973) & 113,126 (161,712) \\
$4\,\fee$ & 1.00 (0.00) & 755,901 (91,748) & 112,825 (161,819) \\
$5\,\fee$ & 1.00 (0.00) & 802,518 (109,834) & 112,407 (161,782) \\
\hline
\end{tabular}
    \caption{Summary statistics for the expected (i) exit time, (ii) fees collected, and (iii) impermanent loss, as $\fee$ varies.  Mean values (with standard deviation in parenthesis) across 10,000 simulations.}
    \label{tab:fee-table}
\end{table}

In the table above, the fees collected increase monotonically with $\fee$. On the bottom half of the table the scaling appears to be linear because on average the exit time is roughly at $T=1$, and the decay from the exponential term is relatively small, however this is not the case in the upper half of the table due to the early exit from the pool. The impermanent loss measure converges in the second half of the table because the LPs do not exit the pool early. From this table we see that the row corresponding to $\fee/2$ has an average exit time of roughly 0.96 (the row below is 1.0 and the one above drops to 0.11). Figure \ref{fig:3dplot-exit} studies this case in more detail, in particular, we show the points at which the LP exits the pool.

\begin{figure}[H]
    \centering
    \includegraphics[width=0.8\linewidth]{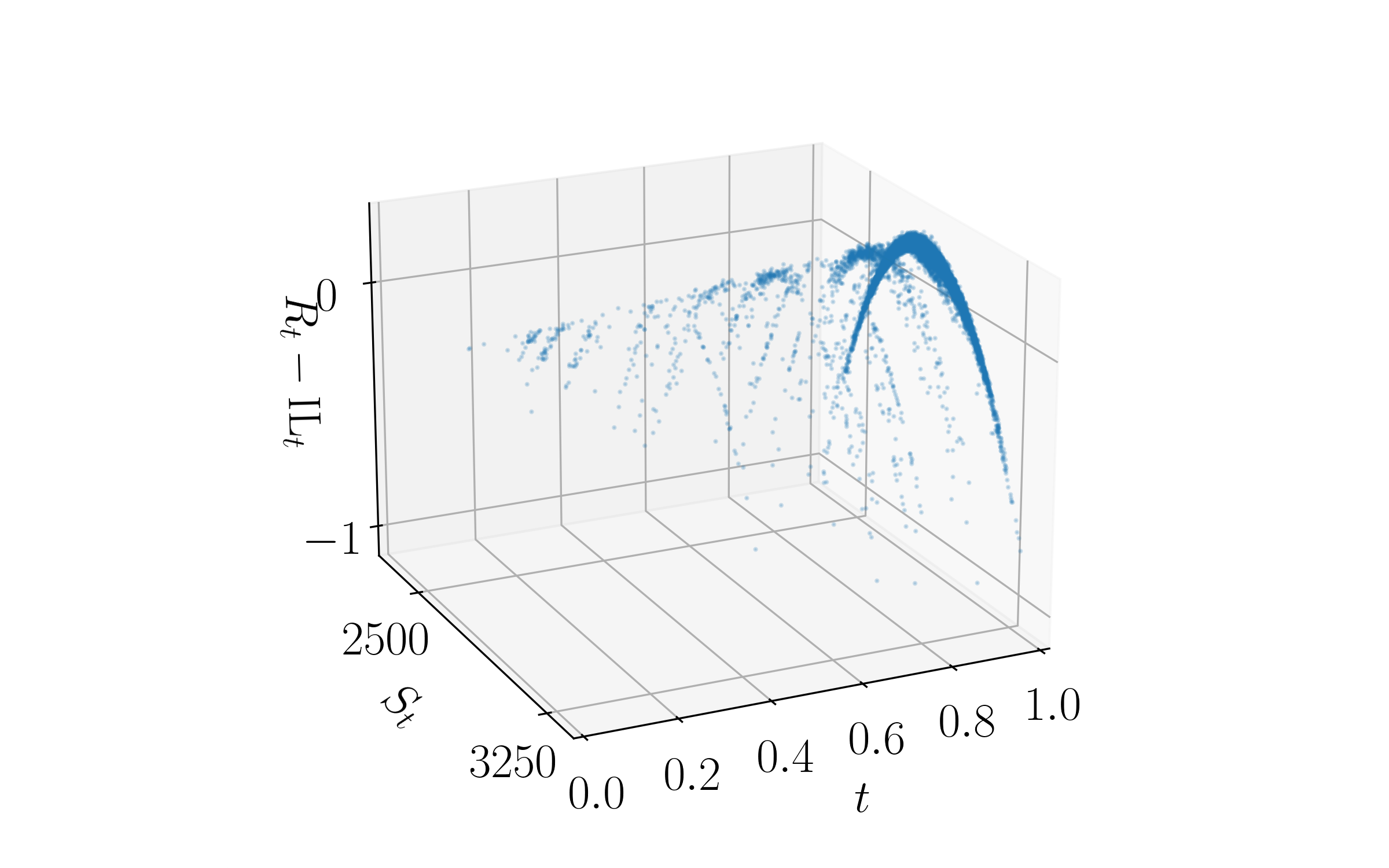}
    \caption{Exit time for the row in Table \ref{tab:fee-table} corresponding to $\fee$/2. The $x$-axis is time, $y$-axis is the oracle price, and $z$-axis is fees minus impermanent loss.  }
    \label{fig:3dplot-exit}
\end{figure}

We observe a quadratic behaviour across various slices of time towards the end of the trading horizon; each dot represents an early exit time. This is due to the impermanent loss taking a dominating role in the performance criterion.  The  plot in Figure \ref{fig:expr-IL} shows the relationship between the fees and $\mathbb{E}[R_\tau - \mathrm{IL}_\tau]$.

\begin{figure}[H]
    \centering
    \includegraphics[width=0.45\linewidth]{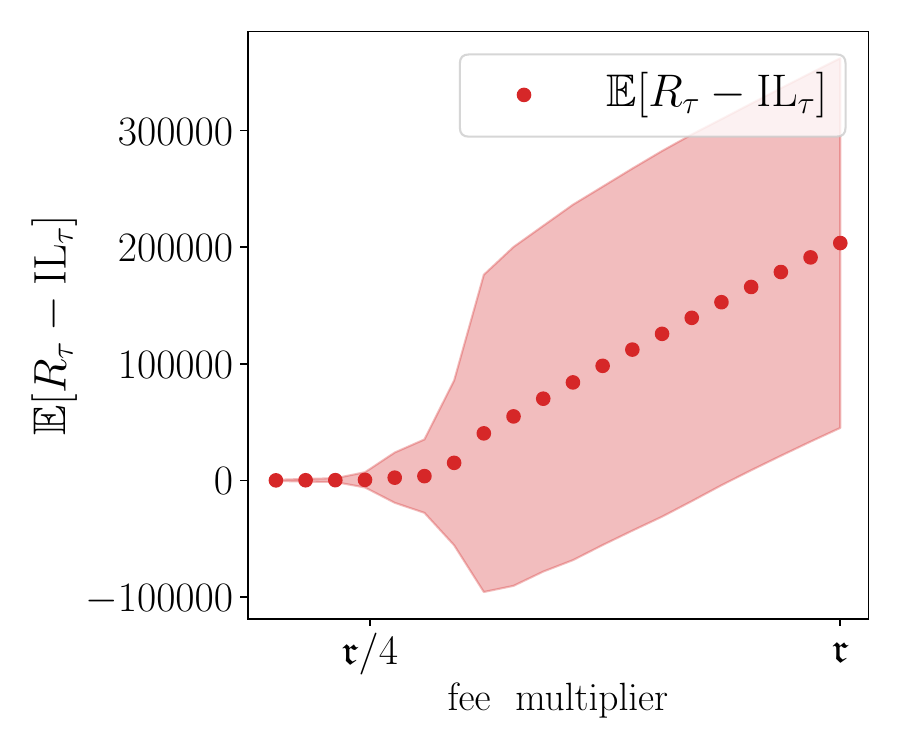}    \includegraphics[width=0.45\linewidth]{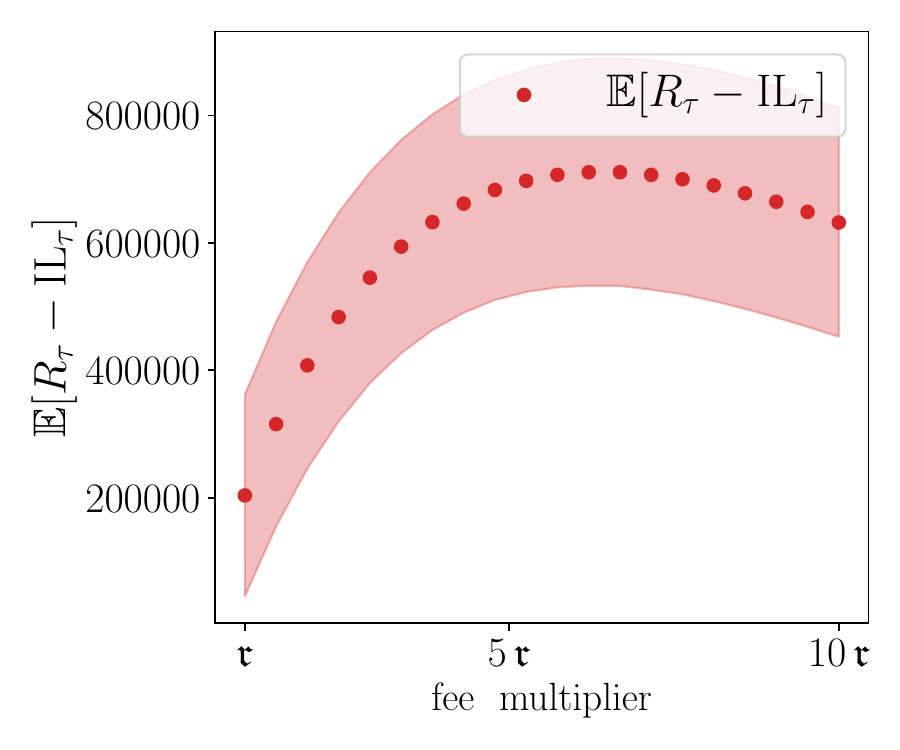}
    \caption{Relationship between $\fee$ and performance, measured through  $\mathbb{E}[R_\tau - \mathrm{IL}_\tau]$. The coloured area shows one standard deviation around the mean. Left panel is zoomed in for $(\fee/10,\fee]$ whereas the right  panel shows $(\fee, 10\,\fee]$. }
    \label{fig:expr-IL}
\end{figure}

From the left panel we see that for some values of $\fee$, the uncertainty may take the performance of the LP to the negative region. Of course, in expectation the performance criterion is always non-negative because the LP has the option to choose $\tau = 0$ as an admissible control, which means that the optimal exit strategy should yield a performance satisfying 
$\mathbb{E}[R_\tau - \mathrm{IL}_\tau]\geq 0$. As the fee increases beyond a given threshold, the decrease in order flow affects the performance criterion of the LP more than the grains from  increasing the values of fee. We also observe an optimal fee multiplier if the objective is that of maximising the criterion of the LPs in the market. Next, in Table \ref{tab:a1-table} we study the effect of noise traders (through $a_1$) and in Table \ref{tab:a2-table} the effect of arbitrageurs (through $a_2$).

\begin{table}[H]
    \centering
    \begin{tabular}{l|rrr}
\hline
\hline
$a_1^{\mathrm{new}}$ & $\mathbb{E}[\tau]$ & $\mathbb{E}[R_\tau]$ & $\mathbb{E}[\mathrm{IL}_\tau]$ \\
\hline
\hline
$a_1$/5 & 1.00 (0.01) & 296,407 (26,305) & 112,619 (157,946) \\
$a_1$/4 & 1.00 (0.01) & 297,518 (26,328) & 112,633 (157,824) \\
$a_1$/3 & 1.00 (0.01) & 299,641 (26,392) & 112,652 (158,199) \\
$a_1$/2 & 1.00 (0.01) & 303,577 (26,514) & 112,625 (157,865) \\
$a_1$ & 1.00 (0.01) & 316,222 (26,930) & 112,635 (157,844) \\
$2\,a_1$ & 1.00 (0.01) & 343,925 (27,871) & 112,757 (158,040) \\
$3\,a_1$ & 1.00 (0.01) & 373,822 (29,214) & 112,910 (158,749) \\
$4\,a_1$ & 1.00 (0.01) & 406,551 (30,439) & 112,767 (159,015) \\
$5\,a_1$ & 1.00 (0.01) & 441,710 (31,530) & 112,592 (158,793) \\
\hline
\end{tabular}
    \caption{Summary statistics for the expected (i) exit time, (ii) fees collected, and (iii) impermanent loss, as $a_1$ varies. Mean values (with standard deviation in parenthesis) across 10,000 simulations. }
    \label{tab:a1-table}
\end{table}

\begin{table}[H]
    \centering
    \begin{tabular}{l|rrr}
\hline
\hline
$a_2^{\mathrm{new}}$ & $\mathbb{E}[\tau]$ & $\mathbb{E}[R_\tau]$ & $\mathbb{E}[\mathrm{IL}_\tau]$ \\
\hline
\hline
$a_2$/5 & 0.86 (0.21) & 118,244 (33,981) & 83,734 (127,349) \\
$a_2$/4 & 0.93 (0.13) & 142,756 (27,971) & 95,493 (139,173) \\
$a_2$/3 & 0.98 (0.07) & 173,770 (24,588) & 104,549 (148,242) \\
$a_2$/2 & 0.99 (0.03) & 216,953 (24,307) & 110,222 (153,239) \\
$a_2$ & 1.00 (0.01) & 316,222 (26,930) & 112,635 (157,844) \\
$2\,a_2$ & 1.00 (0.01) & 476,493 (31,957) & 113,284 (158,588) \\
$3\,a_2$ & 1.00 (0.01) & 618,447 (36,240) & 113,564 (159,533) \\
$4\,a_2$ & 1.00 (0.01) & 750,774 (40,318) & 113,758 (160,513) \\
$5\,a_2$ & 1.00 (0.00) & 878,099 (43,924) & 113,789 (160,830) \\
\hline
\end{tabular}
    \caption{Summary statistics for the expected (i) exit time, (ii) fees collected, and (iii) impermanent loss, as $a_2$ varies. Mean values (with standard deviation in parenthesis) across 10,000 simulations.}
    \label{tab:a2-table}
\end{table}

As expected, the fees collected are monotonically increasing with respect to both, noise traders and arbitrageurs. The key difference is that across all $a_1$ values we consider, the LPs stay in the pool roughly until the end, whereas as we decrease the value of $a_2$ (arbitrageurs), LPs exit the pool early because prices misalign further. Another difference is that with more noise traders the impermanent loss does not change (on average), whereas when there are more arbitrageurs,  LPs realise more of their impermanent loss  up to the maximum level implied by the volatility of the oracle price. \\

\section{Conclusion}\label{sec: conclusions}
We studied the optimal exit problem faced by a liquidity provider in a constant function market, formalised as an optimal stopping problem. Our theoretical analysis characterises the value function as the unique viscosity solution to a Hamilton–Jacobi–Bellman quasi-variational inequality (HJB QVI), and we proposed two numerical approaches to solve the HJB QVI.\\

We found how the LP's optimal exit strategy is driven by the interplay between impermanent loss and fees collected. LPs stay in the pool when the fee income outweighs the potential impermanent loss, which occurs more frequently when volatility is low, fees are high,  or trading activity is high. However, when price dislocations between the AMM price and the oracle price become too large, LPs optimally exit the pool before arbitrageurs realign prices.\\

The value function is maximised when the AMM price equals the oracle price, and declines as the two diverge. In our model, arbitrage activity is not instantaneous: impermanent loss is only realised once arbitrageurs trade to realign prices. Hence, when large price moves occur, LPs may exit the pool pre-emptively to avoid bearing this loss, effectively ``beating'' the arbitrageurs to the exit. While this insight is consistent with the model's dynamics, it is important to note that gas fees would play a key role in this race. In practice, for arbitrageurs to act faster than LPs, they would likely need to pay higher priority fees. The impact of priority fees and more detailed transaction cost modelling is a promising direction for future research, which we are currently investigating.

\appendix

\section{Risk-averse liquidity provider}\label{sec:app_risk_averse} 

Here we provide details for  the problem of a risk averse LP wishing to maximise the expected utility of her PnL. More precisely, we study the problem
\begin{equation}\label{pb_risk_averse}
    \sup_{\tau \in \Tc } \mathbb{E}\left[- \exp \left\{ - \psi \left( P^X_\tau +S_\tau P^Y_\tau + R_\tau \right) \right\}\right]\,,
\end{equation}
where $\psi>0$ represents the absolute risk aversion of the LP. We introduce the process $(P_t)_{t\in \mfT}$ such that for all $t\in \mfT$,
\begin{align*}
    P_t = P^X_t +S_t P^Y_t + R_t.
\end{align*}
In particular, using $\beta^{a,b}$ from  \eqref{eq:betas}, we obtain
\begin{align*}
    \d P_t = \beta^b(Y_{t-},S_t) \d N_t^b  +  \beta^a(Y_{t-},   S_t) \d N_t^a.
\end{align*}
We define the value function $u_\psi$ associated with \eqref{pb_risk_averse}, as

\begin{align}\label{val_func_risk_averse}
\begin{aligned}
u_\psi :\quad & \mfT \times Q \times \mathbb{R} \times \mathbb{R} \longrightarrow \mathbb{R} \\
& (t, y, S, p) \longmapsto \sup_{\tau \in \Tc_{t,T}} \mathbb{E}\Bigg[- \exp \bigg\{ - \psi P^{t,y,S,p}_\tau \bigg\}\Bigg]\,\\
& \qquad \qquad \qquad =\sup_{\tau \in \Tc_{t,T}}  \mathbb{E}\Bigg[- \exp \bigg\{ - \psi \bigg ( p + \displaystyle\int_t^\tau \beta^b(Y^{t,y}_{s\shortminus},S^{t,S}_s) \d N_s^b  + \displaystyle\int_t^\tau \beta^a(Y^{t,y}_{s\shortminus},S^{t,S}_s) \d N_s^a \bigg)\bigg\}\Bigg]\,.
\end{aligned}
\end{align}

The dynamic programming principle for this problem is given in the following lemma.

\begin{lemma}
     Let $(t,y,S,p) \in [0,T) \times Q \times \mathbb{R} \times \mathbb{R}$, then for all $\theta \in \mathcal{T}_{t,T}$ we have
    \begin{align}
        u_\psi(t, y, S,p) = & \, \sup_{\tau \in \Tc_{t,T}} \mathbb{E}\bigg[u_\psi \left(\theta, Y^{t,y}_{\theta }, S^{t,S}_{\theta}, P^{t,p,y,S}_\theta \right) \mathds 1_{\{\theta \le \tau \}} - \exp \left\{-\psi P^{t,y,S,p}_\tau  \right\}\mathds 1_{\{\theta > \tau \}} \bigg] \nonumber.
    \end{align}
\end{lemma}

In particular, the HJB QVI associated with this problem is
\begin{align}\label{eq:hjb_qvi_risk}
        0 = \min\bigg\{& -\frac{\partial}{\partial t} u_\psi(t, y, S, p) \, - \frac{1}{2} \sigma^2 \frac{\partial^2}{\partial S^2} u_\psi(t, y, S,p)\\
        & - \mathds 1_{\{y +\xi \le \overline Y\}}\bar{\lambda}^b(y, S) \big[  u_\psi (t , y +\xi, S, p+\beta^b(y,S)) - u_\psi (t, y, S, p) \big] \, \nonumber \\
    & - \mathds 1_{\{y -\xi \ge \underline Y\}} \bar{\lambda}^a(y, S) \big[  u_\psi(t,y - \xi, S, p+\beta^a(y,S)) - u_\psi(t, y, S,p) \big]\, \nonumber ,\\
    & u_\psi(t,y,S,p) + e^{-\psi p} \bigg\} \qquad \qquad \qquad \qquad \qquad \qquad \qquad \text{on } [0,T) \times  Q \times \mathbb R\times \mathbb R \nonumber,
\end{align}
with terminal condition $u_\psi(T,y,S, p)=-e^{-\psi p}$ for all $(y,S, p) \in Q \times \mathbb R\times \mathbb R$.\\

Next, we use the ansatz
$$u_\psi (t,y,S,p) = -\exp \left\{ - \psi \left( p + v_\psi (t,y,S) \right) \right\},
$$
to obtain the  following HJB QVI for the function $v_\psi$
\begin{align}\label{eq:hjb_qvi_risk_ansatz}
        0 = \min\Bigg\{& -\frac{\partial}{\partial t} v_\psi(t, y, S) \, - \frac{1}{2} \sigma^2 \left(\frac{\partial^2}{\partial S^2} v_\psi(t, y, S) - \psi \left(\frac{\partial}{\partial S} v_\psi(t, y, S) \right)^2  \right) \\
    &\quad - \mathds 1_{\{y +\xi \le \overline Y\}}\bar{\lambda}^b(y, S) \frac 1{\psi} \bigg[  1 - \exp \Big\{ - \psi \left(\beta^b(y,S) + v_\psi(t, y+\xi, S)- v_\psi(t, y, S) \right) \Big\} \bigg] \, \nonumber \\
    &\quad - \mathds 1_{\{y -\xi \ge \underline Y\}} \bar{\lambda}^a(y, S) \frac 1{\psi} \bigg[  1 - \exp \Big\{ - \psi \left(\beta^a(y,S) + v_\psi(t, y-\xi, S)- v_\psi(t, y, S) \right) \Big\} \bigg]\, \nonumber ,\\
    &\quad v_\psi(t,y,S) \Bigg\} \qquad \qquad \qquad \qquad \qquad \qquad \qquad \text{on } [0,T) \times  Q \times \mathbb R\nonumber,
\end{align}
with terminal condition $v_\psi(T,y,S)=0$ for all $(y,S) \in Q \times \mathbb R$.\\

\begin{remark}
    Notice that, as expected,  taking $\psi \rightarrow 0$ in the HJB QVI \eqref{eq:hjb_qvi_risk_ansatz} we obtain the original risk-neutral HJB QVI \eqref{eq:hjb_qvi}.
\end{remark}

Figure~\ref{fig:3dplot-euler-risk-averse} highlights the difference between a risk-neutral and a risk-averse liquidity provider. When the liquidity provider is risk-averse, the value function shrinks, indicating that she tolerates much smaller deviations of the AMM price from the external price. To obtain a stable and meaningful comparison, the parameters are chosen to match those used in Figure~\ref{fig:3dplot-euler-vs-ls}, with $\psi = 10^{-2}$.

\begin{figure}[H]
    \centering
    \includegraphics[width=0.8\linewidth]{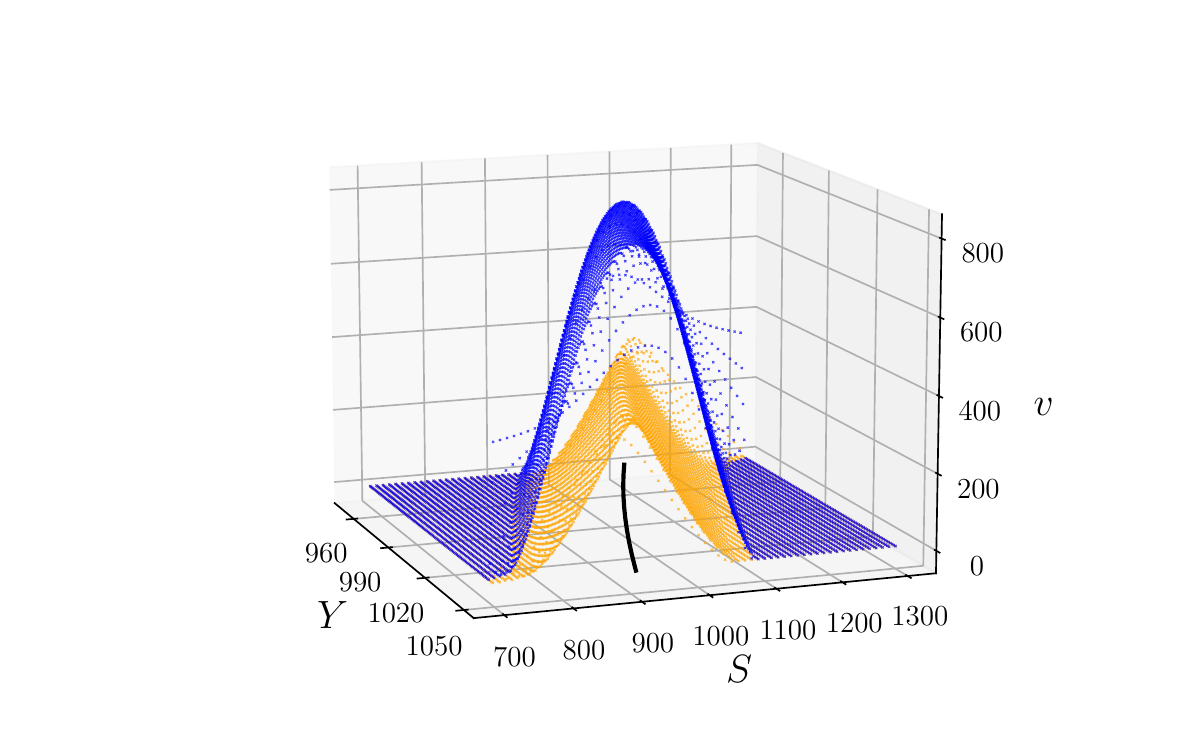}
    \caption{Value function $v$ at $t = 0$ computed using the Euler method for both a risk-neutral (blue surface) and a risk-averse (orange surface) liquidity provider. The black curve in the $S\text{-}Y$ plane represents the states satisfying $S = c / Y^2$.}
    \label{fig:3dplot-euler-risk-averse}
\end{figure}

\section{Proof of Theorem \ref{thm_viscosity}}\label{sec:proof_viscosity}

We report in this section the proof of Theorem \ref{thm_viscosity}. In the first part, we prove that the value function \eqref{val_func_risk_neutral} is a viscosity solution to the HJB QVI \eqref{eq:hjb_qvi} with terminal condition $v(T,y,S) = 0$ for $(y,S)\in Q\times \mathbb R$. In the second part, we prove a comparison principle for the HJB QVI $\eqref{eq:hjb_qvi}$ that allows us to establish  the uniqueness and the continuity of the solution.

\subsection{Existence result}

We first prove that the value function is both a viscosity subsolution and supersolution to the HJB QVI $\eqref{eq:hjb_qvi}$ on $[0,T)\times Q \times \mathbb R$, and then we show that it satisfies the right terminal condition.

\begin{proposition} \label{v_viscosity_subsol}
The function $v$ in \eqref{val_func_risk_neutral} is a viscosity subsolution to the HJB~QVI~\eqref{eq:hjb_qvi} on $[0,T) \times Q \times \mathbb{R}$.
\end{proposition}

\begin{proof}
Using Proposition~\ref{v_bounded}, $v$ is locally bounded on $ [0,T) \times Q \times \mathbb{R}$, so we let  $v^*$ be its upper semicontinuous envelope.\\

Let $(\tilde{t},\tilde{y},\tilde{S}) \in [0,T) \times Q \times \mathbb{R}$ and $\gamma \in \mathcal{C}$ such that $0 = (v^*-\gamma)(\tilde{t},\tilde{y},\tilde{S}) = \max_{(t,y,S) \in [0,T) \times Q \times \mathbb{R}}(v^*-\gamma)(t,y,S)$. As usual, without loss of generality we can assume the maximum to be strict; see for example Lemma 6.1 in \cite{fleming2006controlled}. By definition of $v^*(\tilde{t}, \tilde{y}, \tilde{S})$, there exists $(t_m, S_m)_m$ a sequence in $[0,T) \times \mathbb{R}$ such that
\begin{align}
(t_m, S_m) \xlongrightarrow[\substack{m \to \infty}]{} (\tilde{t}, \tilde{S}), \, \nonumber \\
v(t_m, \tilde{y}, S_m) \xlongrightarrow[\substack{m \to \infty}]{} v^{*}(\tilde{t}, \tilde{y}, \tilde{S}) \nonumber.
\end{align}

We prove the result by contradiction. Assume there exists $\eta > 0$ such that
\begin{align}
\min \Bigg\{ 
&-\frac{\partial}{\partial t} \gamma(\tilde{t}, \tilde{y}, \tilde{S}) - \frac{1}{2} \sigma^2 \frac{\partial^2}{\partial S^2} \gamma(\tilde{t}, \tilde{y}, \tilde{S}) \, \nonumber \\
& - \mathds 1_{\{\tilde y +\xi \le \overline Y\}}\bar{\lambda}^b(\tilde{y}, \tilde{S}) \left[ \beta^b(\tilde{y},\tilde{S}) + \gamma(\tilde{t}, \tilde{y} + \xi, \tilde{S}) - \gamma(\tilde{t}, \tilde{y}, \tilde{S}) \right] \, \nonumber \\ 
& -\mathds 1_{\{\tilde y -\xi \ge \underline Y\}} \bar{\lambda}^a(\tilde{y}, \tilde{S}) \left[ \beta^a(\tilde{y},\tilde{S}) + \gamma(\tilde{t}, \tilde{y} - \xi, \tilde{S}) - \gamma(\tilde{t}, \tilde{y}, \tilde{S}) \right], \, \gamma(\tilde{t}, \tilde{y}, \tilde{S}) \Bigg\} > 2\eta. \nonumber
\end{align}

Then, as $\gamma$ belongs to $\mathcal{C}$, we have
\begin{align}\label{inequality_test_func_sub}
\min \Bigg\{ 
&- \frac{\partial}{\partial t} \gamma(t, y, S) - \frac{1}{2} \sigma^2 \frac{\partial^2}{\partial S^2} \gamma(t, y, S) \, \nonumber \\
& - \mathds 1_{\{y +\xi \le \overline Y\}}\bar{\lambda}^b(y, S) \left[ \beta^b(y,S) + \gamma(t, y + \xi, S) - \gamma(t, y, S) \right] \, \\
& - \mathds 1_{\{y -\xi \ge \underline Y\}}\bar{\lambda}^a(y, S) \left[ \beta^a(y,S) + \gamma(t, y - \xi, S) - \gamma(t, y, S) \right] ,\, \gamma(t, y, S) \Bigg\} \geq \, \eta\nonumber
\end{align}
on $B:=B\left\{ (t,\tilde y, S) \left| |t-\tilde t| < \varepsilon, |S-\tilde S|<\varepsilon \right. \right\}$ for $\varepsilon>0$ small enough. Without loss of generality, we  assume that $B$ contains the sequence $(t_m,  S_m)_m$.
Then, by potentially reducing the value of $\eta$, we have $$v \leq v^* \leq \gamma - \eta$$ on the parabolic boundary $\partial_pB$ of $B$, i.e.,
\begin{align*}
    \partial_pB = &\bigg( \Big( (\tilde t- \varepsilon, \tilde t+\varepsilon) \cap [0,T) \Big) \times\{\tilde y \} \times \{\tilde{S}-\varepsilon , \,\tilde{S}+\varepsilon\} \bigg) \cup \left(\{\tilde{t}+\varepsilon\} \times \{\tilde y \} \times \overline{(\tilde S - \varepsilon, \tilde S +\varepsilon)}  \right).
\end{align*}

Without loss of generality, we can assume that \eqref{inequality_test_func_sub} holds on 
\begin{align}
    \tilde{B} = \{(t, y\pm\xi,S)~|~(t,y,S) \in B \} \nonumber
\end{align}

which is also bounded.\\

For $m\in\N$ we introduce a stopping time $\pi_m = \inf \left\{ t \geq t_m \mid (t, Y_t^{t_m, \tilde{y}}, S_t^{t_m, S_m}) \notin B \right\}$ and by It\^o's formula applied to $\gamma$ between $t_m$ and  $\pi_m \wedge \tau$ where $\tau \in \Tc_{t_m,T}$, we have
\begin{align}
&\gamma(\pi_m \wedge \tau, Y^{t_m, \tilde{y}}_{\pi_m \wedge \tau}, S^{t_m, S_m}_{\pi_m \wedge \tau}) \\
&\qquad =  \, \gamma(t_m, \tilde{y}, S_m) + \int_{t_m}^{\pi_m \wedge \tau} \frac{\partial}{\partial t}\gamma(u, Y^{t_m, \tilde{y}}_{u\shortminus}, S^{t_m, S_m}_u) + \frac{1}{2} \sigma^2 \frac{\partial^2}{\partial S^2}\gamma(u, Y^{t_m, \tilde{y}}_{u\shortminus}, S^{t_m, S_m}_u) \, \d u \, \nonumber \\
&\qquad\quad + \int_{t_m}^{\pi_m \wedge \tau} \sigma \frac{\partial}{\partial S}\gamma(u, Y^{t_m, \tilde{y}}_{u\shortminus}, S^{t_m, S_m}_u) \, \d W_u \, \nonumber \\
&\qquad\quad +\int_{t_m}^{\pi_m \wedge \tau} \left[\gamma(u, Y^{t_m, \tilde{y}}_{u\shortminus} + \xi, S^{t_m, S_m}_u) -\gamma(u, Y^{t_m, \tilde{y}}_{u\shortminus}, S^{t_m, S_m}_u)\right] \, \d N^b_u \, \nonumber \\
&\qquad\quad + \int_{t_m}^{\pi_m \wedge \tau}\left[\gamma(u, Y^{t_m, \tilde{y}}_{u\shortminus} - \xi, S^{t_m, S_m}_u) -\gamma(u, Y^{t_m, \tilde{y}}_{u\shortminus}, S^{t_m, S_m}_u)\right] \, \d N^a_u \, \nonumber,\\
&\qquad = \, \gamma(t_m, \tilde{y}, S_m)  + \int_{t_m}^{\pi_m \wedge \tau} \bigg[\frac{\partial}{\partial t} \gamma(u, Y^{t_m,\tilde{y}}_{u\shortminus}, S^{t_m, S_m}_u) + \frac{1}{2} \sigma^2 \frac{\partial^2}{\partial S^2} \gamma(u, Y^{t_m, \tilde{y}}_{u\shortminus}, S^{t_m, S_m}_u) \, \nonumber \\
&\qquad\quad  \quad + \mathds 1_{\{Y_{u\shortminus} +\xi \le \overline Y\}} \, \bar{\lambda}^b(Y^{t_m, \tilde{y}}_{u\shortminus}, S^{t_m, S_m}_u) \, \left[\gamma(u, Y^{t_m, \tilde{y}}_{u\shortminus} + \xi, S^{t_m, S_m}_u) - \gamma(u, Y^{t_m, \tilde{y}}_{u\shortminus}, S^{t_m, S_m}_u)\right] \, \nonumber \\
&\qquad\quad  \quad + \mathds 1_{\{Y_{u\shortminus} -\xi \ge \underline{Y}\}} \, \bar{\lambda}^a(Y^{t_m, \tilde{y}}_{u\shortminus}, S^{t_m, S_m}_u) \, \left[\gamma(u, Y^{t_m, \tilde{y}}_{u\shortminus} - \xi, S^{t_m, S_m}_u) - \gamma(u, Y_{u\shortminus}, S^{t_m, S_m}_u)\right] \bigg] \d u \, \nonumber \\
&\qquad\quad + \int_{t_m}^{\pi_m \wedge \tau} \left[\gamma(u, Y^{t_m, \tilde{y}}_{u\shortminus} + \xi, S^{t_m, S_m}_u) - \gamma(u, Y^{t_m, \tilde{y}}_{u\shortminus}, S^{t_m, S_m}_u)\right] \, \d\tilde{N}^b_u \,\nonumber \\
&\qquad\quad  + \int_{t_m}^{\pi_m \wedge \tau} \left[\gamma(u, Y^{t_m, \tilde{y}}_{u\shortminus} - \xi, S^{t_m, S_m}_u) - \gamma(u, Y^{t_m, \tilde{y}}_{u\shortminus}, S^{t_m, S_m}_u)\right] \, \d\tilde{N}^a_u \, \nonumber \\
&\qquad\quad   + \int_{t_m}^{\pi_m \wedge \tau} \sigma \frac{\partial}{\partial S} \gamma(u, Y^{t_m, \tilde{y}}_{u\shortminus}, S^{t_m, S_m}_u) \, \d W_u, \nonumber
\end{align}
which we  write as
\begin{align}
\gamma(\pi_m \wedge \tau, Y^{t_m, \tilde{y}}_{\pi_m \wedge \tau}, S^{t_m, S_m}_{\pi_m \wedge \tau}) & \, = \gamma(t_m, \tilde{y}, S_m) \, \nonumber \\
& \hspace{-35mm} - \int_{t_m}^{\pi_m \wedge \tau} - \bigg\{\frac{\partial}{\partial t} \gamma(u, Y^{t_m, \tilde{y}}_{u\shortminus}, S^{t_m, S_m}_u) + \frac{1}{2} \sigma^2 \frac{\partial^2}{\partial S^2} \gamma(u, Y^{t_m, \tilde{y}}_{u\shortminus}, S^{t_m, S_m}_u) \nonumber\\
&\hspace{-35mm} \,\,\, + \mathds 1_{\{Y_{u\shortminus} +\xi \le \overline Y\}} \, \bar{\lambda}^b(Y^{t_m, \tilde{y}}_{u\shortminus}, S^{t_m, S_m}_u) \left[\beta^b(Y^{t_m, \tilde{y}}_{u\shortminus},S^{t_m, S_m}_u) + \gamma(u, Y^{t_m, \tilde{y}}_{u\shortminus} + \xi, S^{t_m, S_m}_u) - \gamma(u, Y^{t_m, \tilde{y}}_{u\shortminus}, S^{t_m, S_m}_u)\right] \nonumber\\
&\hspace{-35mm} \,\,\, + \mathds 1_{\{Y_{u\shortminus} -\xi \ge \underline{Y}\}} \, \bar{\lambda}^a(Y^{t_m, \tilde{y}}_{u\shortminus}, S^{t_m, S_m}_u) \left[\beta^a(Y^{t_m, \tilde{y}}_{u\shortminus},S^{t_m, S_m}_u) + \gamma(u, Y^{t_m, \tilde{y}}_{u\shortminus} - \xi, S^{t_m, S_m}_u) - \gamma(u, Y^{t_m, \tilde{y}}_{u\shortminus}, S^{t_m, S_m}_u)\right] \bigg\} \d u \nonumber\\
&\hspace{-35mm} - \int_{t_m}^{\pi_m \wedge \tau}  \bigg\{\beta^b(Y^{t_m, \tilde{y}}_{u\shortminus},S^{t_m, S_m}_u) \bar{\lambda}^b(Y^{t_m, \tilde{y}}_{u\shortminus}, S^{t_m, S_m}_u) \, \mathds 1_{\{Y_{t\shortminus} +\xi \le \overline Y\}} \,\nonumber \\
&\hspace{-35mm} \quad \, + \beta^a(Y^{t_m, \tilde{y}}_{u\shortminus},S^{t_m, S_m}_u) \bar{\lambda}^a(Y^{t_m, \tilde{y}}_{u\shortminus}, S^{t_m, S_m}_u) \, \mathds 1_{\{Y_{t\shortminus} -\xi \ge \underline{Y}\}}\bigg\} \, \d u \nonumber\\
&\hspace{-35mm} + \int_{t_m}^{\pi_m \wedge \tau} \left[\gamma(u, Y^{t_m, \tilde{y}}_{u\shortminus} + \xi, S^{t_m, S_m}_u) - \gamma(u, Y^{t_m, \tilde{y}}_{u\shortminus}, S^{t_m, S_m}_u)\right] \, \d\tilde{N}^b_u \nonumber\\
&\hspace{-35mm} + \int_{t_m}^{\pi_m \wedge \tau} \left[\gamma(u, Y^{t_m, \tilde{y}}_{u\shortminus} - \xi, S^{t_m, S_m}_u) - \gamma(u, Y^{t_m, \tilde{y}}_{u\shortminus}, S^{t_m, S_m}_u)\right] \, \d\tilde{N}^a_u \nonumber \\
&\hspace{-35mm} + \int_{t_m}^{\pi_m \wedge \tau} \sigma \frac{\partial}{\partial S} \gamma(u, Y^{t_m, \tilde{y}}_{u\shortminus}, S^{t_m, S_m}_u) \, \d W_u. \nonumber
\end{align}

Inside the first integrand, we recognise the left part inside the minimum of the HJB QVI, which we have assumed to be positive on $B$. Moreover, the last three terms are martingales under the probability measure $\mathbb{P}$, because $\gamma$ and its first partial derivative with respect to $S$ are bounded.\\

Therefore, taking the expectation of the above quantity leads to

\begin{align}
\mathbb{E}\left[\gamma(\pi_m \wedge \tau, Y^{t_m, \tilde{y}}_{\pi_m \wedge \tau}, S^{t_m, S_m}_{\pi_m \wedge \tau})\right] \nonumber
& \, \leq \gamma(t_m, \tilde{y}, S_m)  \, \nonumber \\
& \qquad - \mathbb{E} \bigg[\int_{t_m}^{\pi_m \wedge \tau}  \bigg[\mathds 1_{\{Y_{u\shortminus} +\xi \le \overline Y\}}\beta^b(Y^{t_m, \tilde{y}}_{u\shortminus},S^{t_m, S_m}_u) \bar{\lambda}^b(Y^{t_m, \tilde{y}}_{u\shortminus}, S^{t_m, S_m}_u) \bigg. \bigg. \, \nonumber \\
& \qquad\qquad\qquad\quad\quad \quad \bigg. \bigg. + \mathds 1_{\{Y_{u\shortminus} -\xi \ge \underline{Y}\}}\beta^a(Y^{t_m, \tilde{y}}_{u\shortminus},S_u) \bar{\lambda}^a(Y^{t_m, \tilde{y}}_{u\shortminus}, S^{t_m, S_m}_u)\bigg] \, \d u \bigg]. \nonumber
\end{align}

We  also have that
\begin{align}
\gamma(t_m, \tilde{y}, S_m) \geq & \, \mathbb{E}\bigg[\gamma(\pi_m \wedge \tau, Y^{t_m, \tilde{y}}_{\pi_m \wedge \tau}, S^{t_m, S_m}_{\pi_m \wedge \tau}) \nonumber \bigg.\\
& \quad + \int_{t_m}^{\pi_m \wedge \tau}  \bigg[\mathds 1_{\{Y_{u\shortminus} +\xi \le \overline Y\}}\beta^b(Y^{t_m, \tilde{y}}_{u\shortminus},S^{t_m, S_m}_u) \bar{\lambda}^b(Y^{t_m, \tilde{y}}_{u\shortminus}, S^{t_m, S_m}_u) \bigg. \bigg. \, \nonumber \\
&\quad\quad\qquad \quad \quad \quad \bigg. \bigg. + \mathds 1_{\{Y_{u\shortminus} -\xi \ge \underline{Y}\}}\beta^a(Y^{t_m, \tilde{y}}_{u\shortminus},S_u) \bar{\lambda}^a(Y^{t_m, \tilde{y}}_{u\shortminus}, S^{t_m, S_m}_u)\bigg] \, \d u \bigg]. \nonumber
\end{align}

Given that $\gamma$ belongs to $\mathcal{C}$, then $\gamma(t_m, \tilde{y}, S_m) \xlongrightarrow[\substack{m \to \infty}]{} \gamma(\tilde{t}, \tilde{y}, \tilde{S}) = v^*(\tilde{t}, \tilde{y}, \tilde{S})$, and we also have that $v(t_m, \tilde{y}, S_m) \xlongrightarrow[\substack{m \to \infty}]{} v^{*}(\tilde{t}, \tilde{y}, \tilde{S})$. Therefore, there exists an $m$ sufficiently large such that $\gamma(t_m,\tilde{y},S_m) - v(t_m, \tilde{y}, S_m) \leq \frac{\eta}{2}$, and it follows that
\begin{align}
v(t_m, \tilde{y}, S_m) \geq & -\frac{\eta}{2} + \, \mathbb{E}\bigg[\gamma(\pi_m \wedge \tau, Y^{t_m, \tilde{y}}_{\pi_m \wedge \tau}, S^{t_m, S_m}_{\pi_m \wedge \tau}) \nonumber \bigg.\\
&\quad \quad \quad \quad + \int_{t_m}^{\pi_m \wedge \tau}  \bigg[\mathds 1_{\{Y_{u\shortminus} +\xi \le \overline Y\}}\beta^b(Y^{t_m, \tilde{y}}_{u\shortminus},S^{t_m, S_m}_u) \bar{\lambda}^b(Y^{t_m, \tilde{y}}_{u\shortminus}, S^{t_m, S_m}_u) \bigg. \bigg. \, \nonumber \\
&\quad \quad\qquad  \quad \quad \quad \quad \quad \quad \bigg. \bigg. + \mathds 1_{\{Y_{u\shortminus} -\xi \ge \underline{Y}\}}\beta^a(Y^{t_m, \tilde{y}}_{u\shortminus},S_u) \bar{\lambda}^a(Y^{t_m, \tilde{y}}_{u\shortminus}, S^{t_m, S_m}_u)\bigg] \, \d u \bigg]. \nonumber
\end{align}

Morevover,
\begin{align}
\gamma(\pi_m \wedge \tau, Y^{t_m, \tilde{y}}_{\pi_m \wedge \tau}, S^{t_m, S_m}_{\pi_m \wedge \tau}) = \underbrace{\gamma(\pi_m, Y^{t_m, \tilde{y}}_{\pi_m}, S^{t_m, S_m}_{\pi_m})}_{\geq v(\pi_m, Y^{t_m, \tilde{y}}_{\pi_m}, S^{t_m, \tilde{y}}_{\pi_m}) + \eta} \mathds{1}_{\{\pi_m < \tau\}} + \underbrace{\gamma(\tau, Y^{t_m, \tilde{y}}_{\tau}, S^{t_m, S_m}_{\tau})}_{\ge \eta}\mathds{1}_{\{\pi_m \geq \tau\}} \nonumber.
\end{align}

Putting all things together we have
\begin{align}
v(t_m, \tilde{y}, S_m) \geq & \, \, \frac{\eta}{2} + \mathbb{E}\left[\int_{t_m}^{\pi_m \wedge \tau}  \left[ \mathds{1}_{\{Y_{u\shortminus} + \xi \geq \underline{Y}\}} \, \beta^b(Y^{t_m,\tilde{y}}_{u\shortminus},S_u) \bar{\lambda}^b(Y^{t_m,\tilde{y}}_{u\shortminus}, S^{t_m,S_m}_u) \right. \right. \nonumber \\
& \quad \qquad \quad \quad \quad \quad \quad \left. + \mathds{1}_{\{Y_{u\shortminus} - \xi \leq \overline{Y}\}} \, \beta^a(Y^{t_m,\tilde{y}}_{u\shortminus}, S^{t_m,S_m}_u) \bar{\lambda}^a(Y^{t_m,\tilde{y}}_{u\shortminus}, S^{t_m,S_m}_u) \right] \d u \nonumber \\
&\hspace{7cm} + v(\pi_m, Y^{t_m, \tilde{y}}_{\pi_m}, S^{t_m, \tilde{y}}_{\pi_m}) \mathds{1}_{\{\pi_m < \tau\}} \bigg] \nonumber.
\end{align}

By taking the supremum over all the stopping time in $\Tc_{t_m,T}$ on the right-hand side, we get
\begin{align}
v(t_m, \tilde{y}, S_m) > & \, \sup_{\tau \in \Tc_{t_m,T}} \mathbb{E}\left[\int_{t_m}^{\pi_m \wedge \tau}  \left[ \mathds{1}_{\{Y_{u\shortminus} + \xi \geq \underline{Y}\}} \, \beta^b(Y^{t_m,\tilde{y}}_{u\shortminus},S^{t_m,S_m}_u) \bar{\lambda}^b(Y^{t_m,\tilde{y}}_{u\shortminus}, S_u) \right. \right. \nonumber \\
& \quad \quad\qquad\qquad  \quad \quad \quad \left. + \mathds{1}_{\{Y_{u\shortminus} - \xi \leq \overline{Y}\}} \, \beta^a(Y^{t_m,\tilde{y}}_{u\shortminus}, S^{t_m,S_m}_u) \bar{\lambda}^a(Y^{t_m,\tilde{y}}_{u\shortminus}, S^{t_m,S_m}_u) \right] \d u \nonumber \\
&\hspace{7cm}  + v(\pi_m, Y^{t_m, \tilde{y}}_{\pi_m}, S^{t_m, \tilde{y}}_{\pi_m}) \mathds{1}_{\{\pi_m < \tau\}} \bigg] \nonumber,
\end{align}
which contradicts the dynamics programming principle.\\

In conclusion, we necessarily have
\begin{align}
\min \Bigg\{ 
&- \frac{\partial}{\partial t} \gamma(t, y, S) - \frac{1}{2} \sigma^2 \frac{\partial^2}{\partial S^2} \gamma(t, y, S) \, \nonumber \\
& - \mathds 1_{\{y +\xi \le \overline Y\}}\bar{\lambda}^b(y, S) \left[ \beta^b(y,S) + \gamma(t, y + \xi, S) - \gamma(t, y, S) \right] \, \nonumber \\
& - \mathds 1_{\{y -\xi \ge \underline Y\}} \bar{\lambda}^a(y, S) \left[ \beta^a(y,S) + \gamma(t, y - \xi, S) - \gamma(t, y, S) \right] ,\, \gamma(t, y, S)  \Bigg\} \leq 0 \nonumber,
\end{align}
and $v$ is a viscosity subsolution to the HJB QVI on $[0,T) \times Q \times \mathbb{R}$.
\end{proof}

\begin{proposition}\label{v_viscosity_supersol}
The function $v$ in \eqref{val_func_risk_neutral} is a viscosity supersolution to the HJB QVI \eqref{eq:hjb_qvi} on $[0,T) \times Q \times \mathbb{R}$.
\end{proposition}

\begin{proof}
Using Proposition~\ref{v_bounded}, $v$ is locally bounded on $ [0,T) \times Q \times \mathbb{R}$, so we define $v_*$ its lower semicontinuous envelope.\\

Let $(\tilde{t},\tilde{y},\tilde{S}) \in [0,T) \times Q \times \mathbb{R}$ and $\gamma \in \mathcal{C}$ such that $0 = (v_*-\gamma)(\tilde{t},\tilde{y},\tilde{S}) = \min_{(t,y,S) \in [0,T) \times Q \times \mathbb{R}}(v_*-\gamma)(t,y,S)$ and  assume the minimum to be strict. By definition of $v_*(\tilde{t}, \tilde{y}, \tilde{S})$, their exists $(t_m,  S_m)_m$ a sequence of $[0,T)  \times \mathbb{R}$ such that

\begin{align}
(t_m,  S_m) \xlongrightarrow[\substack{m \to \infty}]{} (\tilde{t},\tilde{S}), \, \nonumber \\
v(t_m, \tilde{y}, S_m) \xlongrightarrow[\substack{m \to \infty}]{} v_{*}(\tilde{t}, \tilde{y}, \tilde{S}) \nonumber.
\end{align}

We prove the result by contradiction. Assume there exists $\eta > 0$ such that
\begin{align} \label{inequality_test_func_super}
\min \Bigg\{ 
&-\frac{\partial}{\partial t} \gamma(\tilde{t}, \tilde{y}, \tilde{S}) - \frac{1}{2} \sigma^2 \frac{\partial^2}{\partial S^2} \gamma(\tilde{t}, \tilde{y}, \tilde{S}) \nonumber \, \\
& - \mathds 1_{\{\tilde y +\xi \le \overline Y\}}\bar{\lambda}^b(\tilde{y}, \tilde{S}) \left[ \beta^b(\tilde{y},\tilde{S}) + \gamma(\tilde{t}, \tilde{y} + \xi, \tilde{S}) - \gamma(\tilde{t}, \tilde{y}, \tilde{S}) \right] \, \\
& -\mathds 1_{\{\tilde y -\xi \le \underline Y\}} \bar{\lambda}^a(\tilde{y}, \tilde{S}) \left[ \beta^a(\tilde{y},\tilde{S}) + \gamma(\tilde{t}, \tilde{y} - \xi, \tilde{S}) - \gamma(\tilde{t}, \tilde{y}, \tilde{S}) \right], \, \gamma(\tilde{t}, \tilde{y}, \tilde{S}) \Bigg\} < -2\eta \nonumber
\end{align}

Since $v$ is non-negative, $v_*$ is also non-negative. This implies that $\gamma(\tilde{t}, \tilde{y}, \tilde{S})$ should be non-negative to satisfy $(v_*-\gamma)(\tilde{t},\tilde{y},\tilde{S})=0$. Therefore,  \eqref{inequality_test_func_super} reduces to
\begin{align}
 -\frac{\partial}{\partial t} \gamma(\tilde{t}, \tilde{y}, \tilde{S}) - \frac{1}{2} \sigma^2 \frac{\partial^2}{\partial S^2} \gamma(\tilde{t}, \tilde{y}, \tilde{S}) & - \mathds 1_{\{\tilde y +\xi \le \overline Y\}}\bar{\lambda}^b(\tilde{y}, \tilde{S}) \left[ \beta^b(\tilde{y},\tilde{S}) + \gamma(\tilde{t}, \tilde{y} + \xi, \tilde{S}) - \gamma(\tilde{t}, \tilde{y}, \tilde{S}) \right]\nonumber \, \\
&  -  \mathds 1_{\{\tilde y -\xi \ge \underline Y\}}\bar{\lambda}^a(\tilde{y}, \tilde{S}) \left[ \beta^a(\tilde{y},\tilde{S}) + \gamma(\tilde{t}, \tilde{y} - \xi, \tilde{S}) - \gamma(\tilde{t}, \tilde{y}, \tilde{S}) \right] \, < -2\eta. \nonumber 
\end{align}

Then, as $\gamma$ belongs to $\mathcal{C}$, we have
\begin{align}\label{inequality_test_func_super_reduced}
 -\frac{\partial}{\partial t} \gamma(t, \tilde{y}, S) - \frac{1}{2} \sigma^2 \frac{\partial^2}{\partial S^2} \gamma(t, \tilde{y}, S) 
& - \mathds 1_{\{\tilde y +\xi \le \overline Y\}}\bar{\lambda}^b(\tilde{y}, S) \left[ \beta^b(\tilde{y},S) + \gamma(t, \tilde{y} + \xi, S) - \gamma(t, \tilde{y}, S) \right] \, \\
& - \mathds 1_{\{\tilde y -\xi \ge \underline Y\}}\bar{\lambda}^a(\tilde{y}, S) \left[ \beta^a(\tilde{y},S) + \gamma(t, \tilde{y} - \xi, S) - \gamma(t, \tilde{y}, S) \right] \, \leq \, -\eta \nonumber
\end{align}
on $B:=B\left\{ (t,\tilde y, S) \left| |t-\tilde t| < \varepsilon, |S-\tilde S|<\varepsilon \right. \right\}$ for $\varepsilon>0$ small enough. Without loss of generality, we  assume that $B$ contains the sequence $(t_m,  S_m)_m$.\\

Then, as the minimum of $v_*-\gamma$ is assumed to be strict at $(\tilde{t}, \tilde{y}, \tilde{S})$,  by potentially reducing the value of $\eta$, we have $$v \geq v_* \geq \gamma + \eta$$ on the parabolic boundary $\partial_pB$ of $B$, i.e.,
\begin{align*}
    \partial_pB = &\bigg( \Big( (\tilde t- \varepsilon, \tilde t+\varepsilon) \cap [0,T) \Big) \times\{\tilde y \} \times \{\tilde{S}-\varepsilon , \,\tilde{S}+\varepsilon\} \bigg) \cup \left(\{\tilde{t}+\varepsilon\} \times \{\tilde y \} \times \overline{(\tilde S - \varepsilon, \tilde S +\varepsilon)}  \right).
\end{align*}
Without loss of generality, we  assume that \eqref{inequality_test_func_super_reduced} holds on 
\begin{align}
    \tilde{B} = \{(t, y\pm\xi,S)~|~(t,y,S) \in B \} \nonumber
\end{align}
which is also bounded.\\

For $m\in\N$ we introduce a stopping time $\pi_m = \inf \left\{ t \geq t_m \mid (t, Y_t^{t_m, \tilde{y}}, S_t^{t_m, S_m}) \notin B \right\}$ and by It\^o's formula applied to $\gamma$ between $t_m$ and  $\pi_m \wedge \tau$ where $\tau = \inf \left\{ t \geq t_m \mid \gamma(t, Y_t^{t_m, \tilde{y}}, S_t^{t_m, S_m}) \leq -\eta \right\}\wedge T \in \Tc_{t_m,T}$, we have
\begin{align}
&\gamma(\pi_m \wedge \tau, Y^{t_m, \tilde{y}}_{\pi_m \wedge \tau}, S^{t_m, S_m}_{\pi_m \wedge \tau}) \\
&\qquad =  \, \gamma(t_m, \tilde{y}, S_m) + \int_{t_m}^{\pi_m \wedge \tau} \frac{\partial}{\partial t}\gamma(u, Y^{t_m, \tilde{y}}_{u\shortminus}, S^{t_m, S_m}_u) + \frac{1}{2} \sigma^2 \frac{\partial^2}{\partial S^2}\gamma(u, Y^{t_m, \tilde{y}}_{u\shortminus}, S^{t_m, S_m}_u) \, \d u \, \nonumber \\
&\qquad\quad + \int_{t_m}^{\pi_m \wedge \tau} \sigma \frac{\partial}{\partial S}\gamma(u, Y^{t_m, \tilde{y}}_{u\shortminus}, S^{t_m, S_m}_u) \, \d W_u \, \nonumber \\
&\qquad\quad +\int_{t_m}^{\pi_m \wedge \tau} \left[\gamma(u, Y^{t_m, \tilde{y}}_{u\shortminus} + \xi, S^{t_m, S_m}_u) -\gamma(u, Y^{t_m, \tilde{y}}_{u\shortminus}, S^{t_m, S_m}_u)\right] \, \d N^b_u \, \nonumber \\
&\qquad\quad + \int_{t_m}^{\pi_m \wedge \tau}\left[\gamma(u, Y^{t_m, \tilde{y}}_{u\shortminus} - \xi, S^{t_m, S_m}_u) -\gamma(u, Y^{t_m, \tilde{y}}_{u\shortminus}, S^{t_m, S_m}_u)\right] \, \d N^a_u \, \nonumber,\\
&\qquad = \, \gamma(t_m, \tilde{y}, S_m)  + \int_{t_m}^{\pi_m \wedge \tau} \bigg[\frac{\partial}{\partial t} \gamma(u, Y^{t_m,\tilde{y}}_{u\shortminus}, S^{t_m, S_m}_u) + \frac{1}{2} \sigma^2 \frac{\partial^2}{\partial S^2} \gamma(u, Y^{t_m, \tilde{y}}_{u\shortminus}, S^{t_m, S_m}_u) \, \nonumber \\
&\qquad\quad  \quad + \mathds 1_{\{Y_{u\shortminus} +\xi \le \overline Y\}} \, \bar{\lambda}^b(Y^{t_m, \tilde{y}}_{u\shortminus}, S^{t_m, S_m}_u) \, \left[\gamma(u, Y^{t_m, \tilde{y}}_{u\shortminus} + \xi, S^{t_m, S_m}_u) - \gamma(u, Y^{t_m, \tilde{y}}_{u\shortminus}, S^{t_m, S_m}_u)\right] \, \nonumber \\
&\qquad\quad  \quad + \mathds 1_{\{Y_{u\shortminus} -\xi \ge \underline{Y}\}} \, \bar{\lambda}^a(Y^{t_m, \tilde{y}}_{u\shortminus}, S^{t_m, S_m}_u) \, \left[\gamma(u, Y^{t_m, \tilde{y}}_{u\shortminus} - \xi, S^{t_m, S_m}_u) - \gamma(u, Y_{u\shortminus}, S^{t_m, S_m}_u)\right] \bigg] \d u \, \nonumber \\
&\qquad\quad + \int_{t_m}^{\pi_m \wedge \tau} \left[\gamma(u, Y^{t_m, \tilde{y}}_{u\shortminus} + \xi, S^{t_m, S_m}_u) - \gamma(u, Y^{t_m, \tilde{y}}_{u\shortminus}, S^{t_m, S_m}_u)\right] \, \d\tilde{N}^b_u \,\nonumber \\
&\qquad\quad  + \int_{t_m}^{\pi_m \wedge \tau} \left[\gamma(u, Y^{t_m, \tilde{y}}_{u\shortminus} - \xi, S^{t_m, S_m}_u) - \gamma(u, Y^{t_m, \tilde{y}}_{u\shortminus}, S^{t_m, S_m}_u)\right] \, \d\tilde{N}^a_u \, \nonumber \\
&\qquad\quad   + \int_{t_m}^{\pi_m \wedge \tau} \sigma \frac{\partial}{\partial S} \gamma(u, Y^{t_m, \tilde{y}}_{u\shortminus}, S^{t_m, S_m}_u) \, \d W_u, \nonumber
\end{align}
which we write as
\begin{align}
\gamma(\pi_m \wedge \tau, Y^{t_m, \tilde{y}}_{\pi_m \wedge \tau}, S^{t_m, S_m}_{\pi_m \wedge \tau}) & \, = \gamma(t_m, \tilde{y}, S_m) \, \nonumber \\
& \hspace{-35mm} - \int_{t_m}^{\pi_m \wedge \tau} - \bigg\{\frac{\partial}{\partial t} \gamma(u, Y^{t_m, \tilde{y}}_{u\shortminus}, S^{t_m, S_m}_u) + \frac{1}{2} \sigma^2 \frac{\partial^2}{\partial S^2} \gamma(u, Y^{t_m, \tilde{y}}_{u\shortminus}, S^{t_m, S_m}_u) \nonumber\\
&\hspace{-35mm} \,\,\, + \mathds 1_{\{Y_{u\shortminus} +\xi \le \overline Y\}} \, \bar{\lambda}^b(Y^{t_m, \tilde{y}}_{u\shortminus}, S^{t_m, S_m}_u) \left[\beta^b(Y^{t_m, \tilde{y}}_{u\shortminus},S^{t_m, S_m}_u) + \gamma(u, Y^{t_m, \tilde{y}}_{u\shortminus} + \xi, S^{t_m, S_m}_u) - \gamma(u, Y^{t_m, \tilde{y}}_{u\shortminus}, S^{t_m, S_m}_u)\right] \nonumber\\
&\hspace{-35mm} \,\,\, + \mathds 1_{\{Y_{u\shortminus} -\xi \ge \underline{Y}\}} \, \bar{\lambda}^a(Y^{t_m, \tilde{y}}_{u\shortminus}, S^{t_m, S_m}_u) \left[\beta^a(Y^{t_m, \tilde{y}}_{u\shortminus},S^{t_m, S_m}_u) + \gamma(u, Y^{t_m, \tilde{y}}_{u\shortminus} - \xi, S^{t_m, S_m}_u) - \gamma(u, Y^{t_m, \tilde{y}}_{u\shortminus}, S^{t_m, S_m}_u)\right] \bigg\} \d u \nonumber\\
&\hspace{-35mm} - \int_{t_m}^{\pi_m \wedge \tau}  \bigg\{\beta^b(Y^{t_m, \tilde{y}}_{u\shortminus},S^{t_m, S_m}_u) \bar{\lambda}^b(Y^{t_m, \tilde{y}}_{u\shortminus}, S^{t_m, S_m}_u) \, \mathds 1_{\{Y_{t\shortminus} +\xi \le \overline Y\}} \,\nonumber \\
&\hspace{-35mm} \quad \, + \beta^a(Y^{t_m, \tilde{y}}_{u\shortminus},S^{t_m, S_m}_u) \bar{\lambda}^a(Y^{t_m, \tilde{y}}_{u\shortminus}, S^{t_m, S_m}_u) \, \mathds 1_{\{Y_{t\shortminus} -\xi \ge \underline{Y}\}}\bigg\} \, \d u \nonumber\\
&\hspace{-35mm} + \int_{t_m}^{\pi_m \wedge \tau} \left[\gamma(u, Y^{t_m, \tilde{y}}_{u\shortminus} + \xi, S^{t_m, S_m}_u) - \gamma(u, Y^{t_m, \tilde{y}}_{u\shortminus}, S^{t_m, S_m}_u)\right] \, \d\tilde{N}^b_u \nonumber\\
&\hspace{-35mm} + \int_{t_m}^{\pi_m \wedge \tau} \left[\gamma(u, Y^{t_m, \tilde{y}}_{u\shortminus} - \xi, S^{t_m, S_m}_u) - \gamma(u, Y^{t_m, \tilde{y}}_{u\shortminus}, S^{t_m, S_m}_u)\right] \, \d\tilde{N}^a_u \nonumber \\
&\hspace{-35mm} + \int_{t_m}^{\pi_m \wedge \tau} \sigma \frac{\partial}{\partial S} \gamma(u, Y^{t_m, \tilde{y}}_{u\shortminus}, S^{t_m, S_m}_u) \, \d W_u. \nonumber
\end{align}

Inside the first integrand, we recognise the left part inside the minimum of the HJB QVI, which we have assumed to be negative on $B$. Moreover, the last three terms are martingales under the probability measure $\mathbb{P}$, because $\gamma$ and its first partial derivative with respect to $S$ are bounded.\\

Therefore, taking the expectation of the above quantity leads to
\begin{align}
\mathbb{E}\left[\gamma(\pi_m \wedge \tau, Y^{t_m, \tilde{y}}_{\pi_m \wedge \tau}, S^{t_m, S_m}_{\pi_m \wedge \tau})\right] \nonumber
& \, \geq \gamma(t_m, \tilde{y}, S_m)  \, \nonumber \\
& \qquad - \mathbb{E} \bigg[\int_{t_m}^{\pi_m \wedge \tau}  \bigg[\mathds 1_{\{Y_{u\shortminus} +\xi \le \overline Y\}}\beta^b(Y^{t_m, \tilde{y}}_{u\shortminus},S^{t_m, S_m}_u) \bar{\lambda}^b(Y^{t_m, \tilde{y}}_{u\shortminus}, S^{t_m, S_m}_u) \bigg. \bigg. \, \nonumber \\
&\qquad \qquad\qquad \quad \quad \quad \bigg. \bigg. + \mathds 1_{\{Y_{u\shortminus} -\xi \ge \underline{Y}\}}\beta^a(Y^{t_m, \tilde{y}}_{u\shortminus},S^{t_m, S_m}_u) \bar{\lambda}^a(Y^{t_m, \tilde{y}}_{u\shortminus}, S^{t_m, S_m}_u)\bigg] \, \d u \bigg]. \nonumber
\end{align}

We write the above  as
\begin{align}
\gamma(t_m, \tilde{y}, S_m) \leq & \, \, \mathbb{E}\bigg[\gamma(\pi_m \wedge \tau, Y^{t_m, \tilde{y}}_{\pi_m \wedge \tau}, S^{t_m, S_m}_{\pi_m \wedge \tau}) \nonumber \bigg.\\
& \quad + \int_{t_m}^{\pi_m \wedge \tau}  \bigg[\mathds 1_{\{Y_{u\shortminus} +\xi \le \overline Y\}}\beta^b(Y^{t_m, \tilde{y}}_{u\shortminus},S^{t_m, S_m}_u) \bar{\lambda}^b(Y^{t_m, \tilde{y}}_{u\shortminus}, S^{t_m, S_m}_u) \bigg. \bigg. \, \nonumber \\
&\quad \qquad\quad \quad \quad \quad \bigg. \bigg. + \mathds 1_{\{Y_{u\shortminus} -\xi \ge \underline{Y}\}}\beta^a(Y^{t_m, \tilde{y}}_{u\shortminus},S^{t_m, S_m}_u) \bar{\lambda}^a(Y^{t_m, \tilde{y}}_{u\shortminus}, S^{t_m, S_m}_u)\bigg] \, \d u \bigg]. \nonumber
\end{align}

Given that $\gamma$ belongs to $\mathcal{C}$, then $\gamma(t_m, \tilde{y}, S_m) \xlongrightarrow[\substack{m \to \infty}]{} \gamma(\tilde{t}, \tilde{y}, \tilde{S}) = v_*(\tilde{t}, \tilde{y}, \tilde{S})$, and we also have that $v(t_m, \tilde{y}, S_m) \xlongrightarrow[\substack{m \to \infty}]{} v_{*}(\tilde{t}, \tilde{y}, \tilde{S})$. Therefore, there exists an $m$ sufficiently large such that $v(t_m, \tilde{y}, S_m) - \gamma(t_m,\tilde{y},S_m) \leq \frac{\eta}{2}$. It follows that 
\begin{align}
v(t_m, \tilde{y}, S_m) \leq & \, \frac{\eta}{2} + \, \mathbb{E}\bigg[\gamma(\pi_m \wedge \tau, Y^{t_m, \tilde{y}}_{\pi_m \wedge \tau}, S^{t_m, S_m}_{\pi_m \wedge \tau}) \nonumber \bigg.\\
&\quad \quad \quad + \int_{t_m}^{\pi_m \wedge \tau}  \bigg[\mathds 1_{\{Y_{u\shortminus} +\xi \le \overline Y\}}\beta^b(Y^{t_m, \tilde{y}}_{u\shortminus},S^{t_m, S_m}_u) \bar{\lambda}^b(Y^{t_m, \tilde{y}}_{u\shortminus}, S^{t_m, S_m}_u) \bigg. \bigg. \, \nonumber \\
&\quad \qquad \quad \quad \quad \quad \quad \quad \bigg. \bigg. + \mathds 1_{\{Y_{u\shortminus} -\xi \ge \underline{Y}\}}\beta^a(Y^{t_m, \tilde{y}}_{u\shortminus},S^{t_m, S_m}_u) \bar{\lambda}^a(Y^{t_m, \tilde{y}}_{u\shortminus}, S^{t_m, S_m}_u)\bigg] \, \d u \bigg], \nonumber
\end{align}
moreover,
\begin{align}
\gamma(\pi_m \wedge \tau, Y^{t_m, \tilde{y}}_{\pi_m \wedge \tau}, S^{t_m, S_m}_{\pi_m \wedge \tau}) = \underbrace{\gamma(\pi_m, Y^{t_m, \tilde{y}}_{\pi_m}, S^{t_m, S_m}_{\pi_m})}_{\leq v(\pi_m, Y^{t_m, \tilde{y}}_{\pi_m}, S^{t_m, \tilde{y}}_{\pi_m}) - \eta} \mathds{1}_{\{\pi_m < \tau\}} + \underbrace{\gamma(\tau, Y^{t_m, \tilde{y}}_{\tau}, S^{t_m, S_m}_{\tau})}_{\leq -\eta}\mathds{1}_{\{\pi_m \geq \tau\}} \nonumber.
\end{align}

Combining the above inequalities we have
\begin{align}
v(t_m, \tilde{y}, S_m) \leq & \, -\frac{\eta}{2} + \mathbb{E}\bigg[\int_{t_m}^{\pi_m \wedge \tau}  \Big[ \mathds{1}_{\{Y_{u\shortminus} + \xi \geq \underline{Y}\}} \, \beta^b(Y^{t_m,\tilde{y}}_{u\shortminus},S_u) \bar{\lambda}^b(Y^{t_m,\tilde{y}}_{u\shortminus}, S^{t_m,S_m}_u) \nonumber \\
&\qquad  \quad \quad \quad \quad \quad \quad \quad  + \mathds{1}_{\{Y_{u\shortminus} - \xi \leq \overline{Y}\}} \, \beta^a(Y^{t_m,\tilde{y}}_{u\shortminus}, S^{t_m,S_m}_u) \bar{\lambda}^a(Y^{t_m,\tilde{y}}_{u\shortminus}, S^{t_m,S_m}_u) \Big] \d u \nonumber \\
& \hspace{7cm} + v(\pi_m, Y^{t_m, \tilde{y}}_{\pi_m}, S^{t_m, \tilde{y}}_{\pi_m}) \mathds{1}_{\{\pi_m < \tau\}} \bigg] \nonumber.
\end{align}

By taking the supremum over all the stopping time in $\Tc_{t_m,T}$ on the right-hand side, we get
\begin{align}
v(t_m, \tilde{y}, S_m) < & \, \sup_{\tau \in \Tc_{t_m,T}} \mathbb{E}\bigg[\int_{t_m}^{\pi_m \wedge \tau}  \Big[ \mathds{1}_{\{Y_{u\shortminus} + \xi \geq \underline{Y}\}} \, \beta^b(Y^{t_m,\tilde{y}}_{u\shortminus},S^{t_m,S_m}_u) \bar{\lambda}^b(Y^{t_m,\tilde{y}}_{u\shortminus}, S_u)  \nonumber \\
& \quad\qquad \quad \quad \quad \quad \quad \quad  + \mathds{1}_{\{Y_{u\shortminus} - \xi \leq \overline{Y}\}} \, \beta^a(Y^{t_m,\tilde{y}}_{u\shortminus}, S^{t_m,S_m}_u) \bar{\lambda}^a(Y^{t_m,\tilde{y}}_{u\shortminus}, S^{t_m,S_m}_u) \Big] \d u \nonumber \\
&  \hspace{7cm} + v(\pi_m, Y^{t_m, \tilde{y}}_{\pi_m}, S^{t_m, \tilde{y}}_{\pi_m}) \mathds{1}_{\{\pi_m < \tau\}} \bigg] \nonumber,
\end{align}
which contradicts the dynamics programming principle.\\

In conclusion, we necessarily have that
\begin{align}
\min \Bigg\{ 
&- \frac{\partial}{\partial t} \gamma(t, y, S) - \frac{1}{2} \sigma^2 \frac{\partial^2}{\partial S^2} \gamma(t, y, S) \, \nonumber \\
& - \mathds 1_{\{ y +\xi \le \overline Y\}}\bar{\lambda}^b(y, S) \left[ \beta^b(y,S) + \gamma(t, y + \xi, S) - \gamma(t, y, S) \right] \, \nonumber \\
& - \mathds 1_{\{ y -\xi \ge \underline Y\}}\bar{\lambda}^a(y, S) \left[ \beta^a(y,S) + \gamma(t, y - \xi, S) - \gamma(t, y, S) \right] ,\, \gamma(t, y, S) \Bigg\} \geq 0 \nonumber,
\end{align}
and it follows that  $v$ is a viscosity supersolution to the HJB QVI on $[0,T) \times Q \times \mathbb{R}$.
\end{proof}

\begin{proposition}\label{prop:terminal}
    For all $(y,S) \in Q \times \mathbb R$, we have $v_*(T,y,S) = v^*(T,y,S) = 0.$
\end{proposition}

\begin{proof}
By Proposition \ref{v_bounded}, we know that $v$ is bounded from below by $0$. Therefore, it is also the case for $v_*$ and $v^*$. In particular, we have $v_*(T,y,S) \ge 0$ and $v^*(T,y,S) \ge 0$ for all $(y,S)\in Q \times \mathbb R$. \\

Moreover, by definition, we have $v_* \le v \le v^*$ on $\mfT\times Q \times \mathbb R$. In particular, since $v(T,y,S) = 0$ for all $(y,S)\in Q \times \mathbb R$, we have $v_*(T,y,S) \le 0$ for all $(y,S)\in Q \times \mathbb R$, and we can conclude that $v_*(T,y,S) = 0$.\\

Finally, we just have to prove that $v^*(T,y,S) \le 0$ for all $(y,S)\in Q \times \mathbb R$. Let $(y,S) \in Q \times \mathbb R$. Let $(t_m)_m$ and $(S_m)_m$ be two sequences, respectively in $[0,T)$ and in $\mathbb R$, such that
$$t_m \underset{m\rightarrow \infty}{\longrightarrow} T, \qquad S_m \underset{m\rightarrow \infty}{\longrightarrow} S$$
and
$$v(t_m, y, S_m) \underset{m\rightarrow \infty}{\longrightarrow} v^*(T, y, S).$$
Let $\varepsilon>0$. We introduce a sequence of  stopping times $(\tau_m)_m$ such that for all $m\ge 0$, $\tau_m \in \mathcal T_{t_m, T}$ and $\tau_m$ is $\varepsilon-$optimal in the sense that
\begin{align}\label{lsc_terminal}
v(t_m,y,S_m)\! \le \!\varepsilon \!+\!  \mathbb{E} \!\Bigg[ \!\displaystyle\int_{t_m}^{\tau_m} \!\bigg\{ \!\beta^b\left(Y^m_{u\shortminus}, S^m_u\right)& \, \bar \lambda^b\left(Y^m_{u\shortminus}, S^m_u\right) \mathds 1_{\{ Y^m_{u\shortminus} +\xi \le \overline{Y} \}} 
\!+ \!\beta^a\left(Y^m_{u\shortminus}, S^m_u\right) \, \bar \lambda^a\left(Y^m_{u\shortminus}, S^m_u\right)  \mathds 1_{\{ Y^m_{u\shortminus} - \xi \ge \underline{Y} \}} \!\bigg\} \d u \Bigg]
\end{align}
where we denoted $Y^m = Y^{t_m, y}$ and $S^m = S^{t_m,S_m}$.\\

Using  that $\tau_m$ takes value in $[t_m, T]$, $Y^m$ takes value in the finite set $Q$, and that the functions $\beta^b,\bar \lambda^b, \beta^a, \bar \lambda^a$ have at most linear growth in their second variable, we see that there exists a $\bar C>0$ such that, almost surely,
\begin{align*}
    &\left| \displaystyle\int_{t_m}^{\tau_m} \bigg\{ \beta^b\left(Y^m_{u\shortminus}, S^m_u\right) \, \bar \lambda^b\left(Y^m_{u\shortminus}, S^m_u\right) \mathds 1_{\{ Y^m_{u\shortminus} +\xi \le \overline{Y} \}} 
+ \beta^a\left(Y^m_{u\shortminus}, S^m_u\right) \, \bar \lambda^a\left(Y^m_{u\shortminus}, S^m_u\right)  \mathds 1_{\{ Y^m_{u\shortminus} - \xi \ge \underline{Y} \}} \bigg\} \d u \right|\\
\le & \, \bar C (T-t_m) + \int_{t_m}^T \bar C \left|S^m_u \right|^2 \d u.
\end{align*}
Using Doob's inequality and potentially increasing $\bar C$, we obtain
\begin{align*}
  &\left| \mathbb{E} \Bigg[ \displaystyle\int_{t_m}^{\tau_m} \bigg\{ \beta^b\left(Y^m_{u\shortminus}, S^m_u\right)\, \bar \lambda^b\left(Y^m_{u\shortminus}, S^m_u\right) \mathds 1_{\{ Y^m_{u\shortminus} +\xi \le \overline{Y} \}} 
+ \beta^a\left(Y^m_{u\shortminus}, S^m_u\right) \, \bar \lambda^a\left(Y^m_{u\shortminus}, S^m_u\right)  \mathds 1_{\{ Y^m_{u\shortminus} - \xi \ge \underline{Y} \}} \bigg\} \d u \Bigg] \right| \\
\le & \, \bar C (T-t_m) \left( 1 + |S_m|^2\right).
\end{align*}
Therefore the right-hand side in \eqref{lsc_terminal} goes to $\varepsilon$ when $m\rightarrow \infty$, but since $v(t_m, y, S_m) \underset{m\rightarrow \infty}{\longrightarrow} v^*(T, y, S)$, we obtain 
$$v^*(T, y, S) \le \varepsilon.$$
We conclude the proof by sending $\varepsilon$ to $0$.
\end{proof}

\subsection{Uniqueness result and continuity}

Let us first make a change of variable before proving the main result. For a given $\rho>0$, we define the function $\tilde v : \mfT\times Q \times \mathbb R \rightarrow \mathbb R$ by
$$\tilde v(t,y,S) = e^{\rho t} v(t,y,s)$$
for all $(t,y,S) \in \mfT\times Q \times \mathbb R$.\\

From what precedes, it is clear that $\tilde v$ is a viscosity solution to the HJB QVI
\begin{align}\label{eq:hjb_qvi_tilde}
        0 = \min\Bigg\{& \rho \tilde v (t,y,S)-\frac{\partial}{\partial t} \tilde v(t, y, S) \, - \frac{1}{2} \sigma^2 \frac{\partial^2}{\partial S^2} \tilde v(t, y, S)\\
        & - \mathds 1_{\{y +\xi \le \overline Y\}}\bar{\lambda}^b(y, S) \big[ \tilde \beta^b(t, y,S) + \tilde v(t, y +\xi, S) - \tilde v(t, y, S) \big] \, \nonumber \\
    & - \mathds 1_{\{y -\xi \ge \underline Y\}} \bar{\lambda}^a(y, S) \big[ \tilde \beta^a(t, y,S) + \tilde v(t, y - \xi, S) - \tilde v(t, y, S) \big]\,  , \, \tilde v(t,y,S) \Bigg\} \qquad \text{on } [0,T) \times  Q \times \mathbb R ,\nonumber
\end{align}
with terminal condition $\tilde v(T,y,S)=0$ for all $(y,S) \in Q \times \mathbb R$, where 
$$\tilde \beta^b(t,y,S) = e^{\rho t} \beta^b(y,S) \qquad \text{and} \qquad \tilde \beta^a(t,y,S) = e^{\rho t} \beta^a(y,S).$$

We then define the concepts of parabolic superdifferential and subdifferential before introducing an equivalent characterization of viscosity solution.

\begin{definition}
Let $u:\mfT\times Q \times \mathbb R \rightarrow \mathbb R$ be a USC function, $w:\mfT\times Q \times \mathbb R \rightarrow \mathbb R$ a LSC function, and let $(\tilde t, \tilde y, \tilde S) \in [0,T) \times Q \times \mathbb R$.
    \begin{enumerate}[label=(\roman*)]
        \item  We say that $(q,p,A)\in \mathbb R^3$ is in $\mathcal P^+u(\tilde t, \tilde y, \tilde S)$ if 
        \begin{align*}
            u(t, \tilde y, S) \le u(\tilde t, \tilde y, \tilde S) + q(t-\tilde t) + p(S- \tilde S) + \frac 12 A(S-\tilde S)^2 + o \left( |t-\tilde t| + |S-\tilde S|^2 \right)
        \end{align*}
        for all $(t,S) \in [0,T) \times \mathbb R.$
        \item  We say that $(q,p,A)\in \mathbb R^3$ is in $\mathcal P^-w(\tilde t, \tilde y, \tilde S)$ if 
        \begin{align*}
            w(t, \tilde y, S) \ge w(\tilde t, \tilde y, \tilde S) + q(t-\tilde t) + p(S- \tilde S) + \frac 12 A(S-\tilde S)^2 + o \left( |t-\tilde t| + |S-\tilde S|^2 \right)
        \end{align*}
        for all $(t,S) \in [0,T) \times \mathbb R.$
        \item We also define $\bar {\mathcal P}^+u(\tilde t, \tilde y, \tilde S)$ as the set of $(q,p,A)\in \mathbb R^3$ such that the exists a sequence $(t_n, S_n)_n$ in $[0,T)\times \mathbb R$ and a sequence $(q_n, p_n, A_n)_n$ in $\mathbb R^3$ such that we have $(q_n, p_n, A_n) \in \mathcal P^+ u(t_n, \tilde y, S_n)$ for all $n\in \mathbb N^*$, and
        $$(q_n, p_n,A_n) \underset{n\rightarrow \infty}{\longrightarrow} (q,p,A).$$
        \item Similarly, $\bar {\mathcal P}^-w(\tilde t, \tilde y, \tilde S)$ is the set of $(q,p,A)\in \mathbb R^3$ such that the exists a sequence $(t_n, S_n)_n$ in $[0,T)\times \mathbb R$ and a sequence $(q_n, p_n, A_n)_n$ in $\mathbb R^3$ such that we have $(q_n, p_n, A_n) \in \mathcal P^- w(t_n, \tilde y, S_n)$ for all $n\in \mathbb N^*$, and
        $$(q_n, p_n,A_n) \underset{n\rightarrow \infty}{\longrightarrow} (q,p,A).$$
    \end{enumerate}
\end{definition}

\begin{lemma}\label{app:equivalent_visc}
    \begin{enumerate}[label=(\roman*)]
        \item Let $\tilde u$ be a USC function on $\mfT\times Q \times \mathbb R$. Then $\tilde u$ is a subsolution to the HJB QVI \eqref{eq:hjb_qvi_tilde} on $[0,T) \times Q \times \mathbb R$ if and only if for all $(\tilde t, \tilde y, \tilde S) \in [0,T) \times Q \times \mathbb R$ and $(\tilde q, \tilde p, \tilde A) \in \bar {\mathcal P}^+u(\tilde t, \tilde y, \tilde S)$, we have
        \begin{align*}
        0 \ge \min\Bigg\{& \rho \tilde u (\tilde t,\tilde y,\tilde S)-\tilde q \, - \frac{1}{2} \sigma^2 \tilde A\\
        & - \mathds 1_{\{\tilde y +\xi \le \overline Y\}}\bar{\lambda}^b(\tilde y, \tilde S) \big[ \tilde \beta^b(\tilde t, \tilde y,\tilde S) + \tilde u(\tilde t, \tilde y +\xi, \tilde S) - \tilde u(\tilde t, \tilde y, \tilde S) \big] \, \nonumber \\
    & - \mathds 1_{\{\tilde y -\xi \ge \underline Y\}} \bar{\lambda}^a(\tilde y, \tilde S) \big[ \tilde \beta^a(\tilde t, \tilde y,\tilde S) + \tilde u(\tilde t, \tilde y - \xi, \tilde S) - \tilde u(\tilde t, \tilde y, \tilde S) \big]\,  , \, \tilde u(\tilde t,\tilde y,\tilde S) \Bigg\}.\nonumber
\end{align*}
        \item Let $\tilde w$ be a LSC function on $\mfT\times Q \times \mathbb R$. Then $\tilde u$ is a supersolution to the HJB QVI \eqref{eq:hjb_qvi_tilde} on $[0,T) \times Q \times \mathbb R$ if and only if for all $(\tilde t, \tilde y, \tilde S) \in [0,T) \times Q \times \mathbb R$ and $(\tilde q, \tilde p, \tilde A) \in \bar {\mathcal P}^-w(\tilde t, \tilde y, \tilde S)$, we have
        \begin{align*}
        0 \le \min\Bigg\{& \rho \tilde w (\tilde t,\tilde y,\tilde S)-\tilde q \, - \frac{1}{2} \sigma^2 \tilde A\\
        & - \mathds 1_{\{\tilde y +\xi \le \overline Y\}}\bar{\lambda}^b(\tilde y, \tilde S) \big[ \tilde \beta^b(\tilde t, \tilde y,\tilde S) + \tilde w(\tilde t, \tilde y +\xi, \tilde S) - \tilde w(\tilde t, \tilde y, \tilde S) \big] \, \nonumber \\
    & - \mathds 1_{\{\tilde y -\xi \ge \underline Y\}} \bar{\lambda}^a(\tilde y, \tilde S) \big[ \tilde \beta^a(\tilde t, \tilde y,\tilde S) + \tilde w(\tilde t, \tilde y - \xi, \tilde S) - \tilde w(\tilde t, \tilde y, \tilde S) \big]\,  , \, \tilde w(\tilde t,\tilde y,\tilde S) \Bigg\}.\nonumber
\end{align*}
    \end{enumerate}
\end{lemma}

This is a classical result, and we refer the reader to Proposition 4.1 in \cite{fleming2006controlled}  for a detailed proof. We can now state the comparison principle. 

\begin{proposition}\label{prop:comparison}
    Let $u \in \Xi$ (resp. $w \in \Xi$) be a USC subsolution (resp. LSC supersolution) to the HJB QVI \eqref{eq:hjb_qvi} with $w\ge u$ on $\{T\}\times Q \times \mathbb R$. Then we have $w\ge u$ on $\mfT\times Q \times \mathbb R$.
\end{proposition}

\begin{proof}
    The proof is technical and we decompose it in several steps.\\

    \textbf{Step 1: Change of variable.} \\

    We introduce the functions $\tilde u, \tilde w \in \Xi$ given by
    $$\tilde u (t,y,S) = e^{\rho t} u(t,y,S) \quad \text{and} \quad \tilde w (t,y,S) = e^{\rho t} w(t,y,S) \quad \forall \, (t,y,S) \in \mfT \times Q \times \mathbb R.$$

    Then $\tilde u$ and $\tilde w$ are respectively subsolution and supersolution to the HJB QVI \eqref{eq:hjb_qvi_tilde} on $[0,T) \times Q\times \mathbb R$ with $\tilde w \ge \tilde u$ on $\{T\}\times Q \times \mathbb R$. It suffices to prove that $\tilde w \ge \tilde u$ on $\mfT\times Q \times \mathbb R$.\\

    \textbf{Step 2: Doubling variables.} \\

    By contradiction, assume $\underset{\mfT\times Q \times \mathbb R}{\sup} \tilde u - \tilde w >0$. Then for $\varepsilon,\mu>0$ with $\varepsilon$ small enough, there exists $(\tilde t, \tilde y, \tilde S) \in \mfT\times Q \times \mathbb R$ such that
    \begin{align}\label{app:contrad0}
        0&< \tilde u (\tilde t, \tilde y, \tilde S) - \tilde w (\tilde t, \tilde y, \tilde S) - \phi(\tilde t, \tilde S, \tilde S)\\
        &= \underset{(t,y,S) \in [0,T)\times Q \times \mathbb R}{\sup} \tilde u(t,y,S) - \tilde w (t,y,S) - \phi (t,S,S),\nonumber
    \end{align}
    where
    $$\phi (t,S,R) = \varepsilon e^{-\mu t} \left(1 + S^4 + R^4 \right)\quad \forall \; (t,S,R) \in \mfT \times \mathbb R^2.$$
    The choice of $\phi$ adds some coercivity and allows to look for the supremum in a bounded set. Moreover, $\tilde t <T$ as $\tilde w\ge \tilde u$ on $\{T\}\times Q \times \mathbb R$.\\

    Then for all $n \in \mathbb N^*$, we can find $(t_n, S_n, R_n) \in \mfT \times \mathbb R^2$ such that 
    \begin{align}\label{app:contrad}
        0&< \tilde u (t_n, \tilde y, S_n) - \tilde w (t_n, \tilde y, R_n) - \phi(t_n, S_n, R_n) - n(S_n - R_n)^2 - \left( (t_n-\tilde t)^2 + (S_n - \tilde S)^4\right)\\
        &= \underset{(t,S,R) \in \mfT\times \mathbb R^2}{\sup} \tilde u (t,\tilde y, S) - \tilde w(t,\tilde y, R) - \phi(t,S,R) - n (S-R)^2 - \left( (t-\tilde t)^2 + (S - \tilde S)^4\right).\nonumber
    \end{align}

    It is clear the $|S_n - R_n| \underset{n\rightarrow \infty}{\longrightarrow} 0$, and up to a subsequence there exists $(\hat t, \hat S) \in \mfT \times \mathbb R$ such that 
    $$t_n  \underset{n\rightarrow \infty}{\longrightarrow} \hat t \qquad \text{and} \qquad S_n, R_n  \underset{n\rightarrow \infty}{\longrightarrow} \hat S.$$

    Moreover, we clearly have 
    \begin{align}\label{app:ineq0}
        \tilde u (\tilde t, \tilde y, \tilde S) & - \tilde w (\tilde t, \tilde y, \tilde S) - \phi(\tilde t, \tilde S, \tilde S)  \\
        \le &\; \tilde u (t_n, \tilde y, S_n) - \tilde w (t_n, \tilde y, R_n) - \phi(t_n, S_n, R_n) - n(S_n - R_n)^2 - \left( (t_n-\tilde t)^2 + (S_n - \tilde S)^4\right),\nonumber
    \end{align}
    so in particular
      \begin{align*}
         \tilde u (\tilde t, \tilde y, \tilde S) - \tilde w (\tilde t, \tilde y, \tilde S) - \phi(\tilde t, \tilde S, \tilde S) 
        \le \tilde u (t_n, \tilde y, S_n) - \tilde w (t_n, \tilde y, R_n) - \phi(t_n, S_n, R_n),
    \end{align*}
    and therefore
    \begin{align*}
         \tilde u (\tilde t, \tilde y, \tilde S) - \tilde w (\tilde t, \tilde y, \tilde S) - \phi(\tilde t, \tilde S, \tilde S)  
        \le  \underset{n\rightarrow \infty}{\lim \inf} \left( \tilde u (t_n, \tilde y, S_n) - \tilde w (t_n, \tilde y, R_n) \right)- \phi(\hat t, \hat S, \hat S).
    \end{align*}
    However, by definition,
    \begin{align*}
        \underset{n\rightarrow \infty}{\lim \sup}\left( \tilde u (t_n, \tilde y, S_n) - \tilde w (t_n, \tilde y, R_n) \right) \le \tilde u (\hat t, \tilde y, \hat S) - \tilde w (\hat t, \tilde y, \hat S),
    \end{align*}
so
 \begin{align*}
        \tilde u (\tilde t, \tilde y, \tilde S) - \tilde w (\tilde t, \tilde y, \tilde S) - \phi(\tilde t, \tilde S, \tilde S)  
        \le  \tilde u (\hat t, \tilde y, \hat S) - \tilde w (\hat t, \tilde y, \hat S) - \phi(\hat t, \hat S, \hat S).
    \end{align*}
    This implies by \eqref{app:contrad0} that
    \begin{align*}
       \underset{n\rightarrow \infty}{\lim}\tilde u (t_n, \tilde y, S_n) - \tilde w (t_n, \tilde y, R_n) - \phi(t_n, S_n, R_n) &= \tilde u (\hat t, \tilde y, \hat S) - \tilde w (\hat t, \tilde y, \hat S) - \phi(\hat t, \hat S, \hat S)\\ &= \tilde u (\tilde t, \tilde y, \tilde S) - \tilde w (\tilde t, \tilde y, \tilde S) - \phi(\tilde t, \tilde S, \tilde S) .
    \end{align*}
    Going back to \eqref{app:contrad0}, we get 
    \begin{align*}
        \tilde u (\tilde t, \tilde y, \tilde S) - \tilde w (\tilde t, \tilde y, \tilde S) - \phi(\tilde t, \tilde S, \tilde S)\le \tilde u (\tilde t, \tilde y, \tilde S) - \tilde w (\tilde t, \tilde y, \tilde S) - \phi(\tilde t, \tilde S, \tilde S) - \left( (\hat t-\tilde t)^2 + (\hat S - \tilde S)^4\right),
    \end{align*}
    which can only hold if $\hat t = \tilde t$ and $\hat S = \tilde S$, and we necessarily have 
    $$\underset{n\rightarrow \infty}{\lim}\tilde u (t_n, \tilde y, S_n) - \tilde w (t_n, \tilde y, R_n) = \tilde u (\tilde t, \tilde y, \tilde S) - \tilde w (\tilde t, \tilde y, \tilde S).$$

    \textbf{Step 3: Ishii's lemma.} \\

    For all $n\in \mathbb N^*$, let us define the function $\varphi_n$ by
    $$\varphi_n(t,S,R) = \phi(t,S,R) + n(S-R)^2 + \left( (t-\tilde t)^2 + (S-\tilde S)^2\right) \qquad \forall (t,S,R) \in \mfT \times \mathbb R^2.$$

    Then, by Ishii's lemma (see for instance \cite{fleming2006controlled} Theorem 6.1), we know that for $\eta>0$ there exist $(q^1_n, p^1_n, A^1_n) \in \bar{\mathcal P}^+ \tilde u (t_n, \tilde y, S_n)$ and $(q^2_n, p^2_n, A^2_n) \in \bar{\mathcal P}^- \tilde w (t_n, \tilde y, R_n)$ such that
    $$q^1_n - q^2_n = \frac{\partial}{\partial t} \varphi_n(t_n, S_n, R_n), \qquad \left(p^1_n, p^2_n \right) = \left( \frac{\partial}{\partial S} \varphi_n(t_n, S_n, R_n), -\frac{\partial}{\partial R} \varphi_n(t_n, S_n, R_n) \right),$$
    and 
    $$\begin{pmatrix}
        A^1_n & 0 \\ 0 & -A^2_n
    \end{pmatrix} \le H \varphi_n(t_n, S_n, R_n) + \eta \left(H \varphi_n(t_n, S_n, R_n) \right)^2,$$
    where $H \varphi_n(t, S, R)$ denotes the Hessian matrix of $\varphi_n$ with respect to $(S,R)$.\\

    \textbf{Step 4: Viscosity properties.} \\

    We can now apply Lemma \ref{app:equivalent_visc} to see that, on the one hand, 
    \begin{align*}
        0 \ge \min\Bigg\{& \rho \tilde u ( t_n,\tilde y,S_n)-q^1_n \, - \frac{1}{2} \sigma^2  A^1_n\\
        & - \mathds 1_{\{\tilde y +\xi \le \overline Y\}}\bar{\lambda}^b(\tilde y, S_n) \big[ \tilde \beta^b(t_n, \tilde y,S_n) + \tilde u(t_n, \tilde y +\xi, S_n) - \tilde u(t_n, \tilde y, S_n) \big] \, \nonumber \\
    & - \mathds 1_{\{\tilde y -\xi \ge \underline Y\}} \bar{\lambda}^a(\tilde y, S_n) \big[ \tilde \beta^a(t_n, \tilde y,S_n) + \tilde u(t_n, \tilde y - \xi, S_n) - \tilde u(t_n, \tilde y, S_n) \big]\,  , \, \tilde u(t_n,\tilde y,S_n) \Bigg\},\nonumber
\end{align*}
but by \eqref{app:contrad} and the fact that $\tilde w \in \Xi$, we know that $\tilde u(t_n, \tilde y, S_n) >0$, so necessarily
\begin{align}\label{app:first_ineq}
    0 \ge &\;\rho \tilde u ( t_n,\tilde y,S_n)-q^1_n \, - \frac{1}{2} \sigma^2  A^1_n\\
        & - \mathds 1_{\{\tilde y +\xi \le \overline Y\}}\bar{\lambda}^b(\tilde y, S_n) \big[ \tilde \beta^b(t_n, \tilde y,S_n) + \tilde u(t_n, \tilde y +\xi, S_n) - \tilde u(t_n, \tilde y, S_n) \big] \, \nonumber \\
    & - \mathds 1_{\{\tilde y -\xi \ge \underline Y\}} \bar{\lambda}^a(\tilde y, S_n) \big[ \tilde \beta^a(t_n, \tilde y,S_n) + \tilde u(t_n, \tilde y - \xi, S_n) - \tilde u(t_n, \tilde y, S_n)\big].\nonumber
\end{align}
On the other hand, still applying Lemma \ref{app:equivalent_visc}, we have 
\begin{align*}
        0 \le \min\Bigg\{& \rho \tilde w ( t_n,\tilde y,R_n)-q^2_n \, - \frac{1}{2} \sigma^2  A^2_n\\
        & - \mathds 1_{\{\tilde y +\xi \le \overline Y\}}\bar{\lambda}^b(\tilde y, R_n) \big[ \tilde \beta^b(t_n, \tilde y,R_n) + \tilde w(t_n, \tilde y +\xi, R_n) - \tilde w(t_n, \tilde y, R_n) \big] \, \nonumber \\
    & - \mathds 1_{\{\tilde y -\xi \ge \underline Y\}} \bar{\lambda}^a(\tilde y, R_n) \big[ \tilde \beta^a(t_n, \tilde y,R_n) + \tilde w(t_n, \tilde y - \xi, R_n) - \tilde w(t_n, \tilde y, R_n) \big]\,  , \, \tilde w(t_n,\tilde y,S_n) \Bigg\},\nonumber
\end{align*}
which implies 
\begin{align}\label{app:scd_ineq}
0\le &\; \rho \tilde w ( t_n,\tilde y,R_n)-q^2_n \, - \frac{1}{2} \sigma^2  A^2_n\\
        & - \mathds 1_{\{\tilde y +\xi \le \overline Y\}}\bar{\lambda}^b(\tilde y, R_n) \big[ \tilde \beta^b(t_n, \tilde y,R_n) + \tilde w(t_n, \tilde y +\xi, R_n) - \tilde w(t_n, \tilde y, R_n) \big] \, \nonumber \\
    & - \mathds 1_{\{\tilde y -\xi \ge \underline Y\}} \bar{\lambda}^a(\tilde y, R_n) \big[ \tilde \beta^a(t_n, \tilde y,R_n) + \tilde w(t_n, \tilde y - \xi, R_n) - \tilde w(t_n, \tilde y, R_n) \big].\nonumber    
\end{align}

From \eqref{app:first_ineq} and \eqref{app:scd_ineq}, we get 
\begin{align*}
    \rho \left( \tilde u ( t_n,\tilde y,S_n) - \tilde w ( t_n,\tilde y,R_n)\right) \le &\;q^1_n - q^2_n + \frac 12 \sigma^2 \left(A^1_n-A^2_n\right)\\
    & + \mathds 1_{\{\tilde y +\xi \le \overline Y\}}\bar{\lambda}^b(\tilde y, S_n) \big[ \tilde \beta^b(t_n, \tilde y,S_n) + \tilde u(t_n, \tilde y +\xi, S_n) - \tilde u(t_n, \tilde y, S_n) \big] \, \nonumber \\
    & + \mathds 1_{\{\tilde y -\xi \ge \underline Y\}} \bar{\lambda}^a(\tilde y, S_n) \big[ \tilde \beta^a(t_n, \tilde y,S_n) + \tilde u(t_n, \tilde y - \xi, S_n) - \tilde u(t_n, \tilde y, S_n)\big]\\
    & - \mathds 1_{\{\tilde y +\xi \le \overline Y\}}\bar{\lambda}^b(\tilde y, R_n) \big[ \tilde \beta^b(t_n, \tilde y,R_n) + \tilde w(t_n, \tilde y +\xi, R_n) - \tilde w(t_n, \tilde y, R_n) \big] \, \nonumber \\
    & - \mathds 1_{\{\tilde y -\xi \ge \underline Y\}} \bar{\lambda}^a(\tilde y, R_n) \big[ \tilde \beta^a(t_n, \tilde y,R_n) + \tilde w(t_n, \tilde y - \xi, R_n) - \tilde w(t_n, \tilde y, R_n) \big].
\end{align*}
Moreover, 
$$H \varphi_n(t_n, S_n, R_n) = 
\begin{pmatrix}
    \frac{\partial^2}{\partial S^2} \phi(t_n, S_n, R_n) +2n + 12(S_n-\tilde S)^2 & \frac{\partial^2}{\partial S \partial R} \phi(t_n, S_n, R_n)  - 2n\\
    \frac{\partial^2}{\partial S \partial R} \phi(t_n, S_n, R_n)  - 2n & \frac{\partial^2}{\partial R^2} \phi(t_n, S_n, R_n) + 2n
\end{pmatrix},$$
and we get
\begin{align*}
    \rho \left( \tilde u ( t_n,\tilde y,S_n) - \tilde w ( t_n,\tilde y,R_n)\right) \le &\;\frac{\partial}{\partial t}\phi(t_n, S_n, R_n) + 2(t_n-\tilde t)+ \eta C_n \\
    & + \frac 12 \sigma^2 \bigg(\frac{\partial^2}{\partial S^2} \phi(t_n, S_n, R_n) + \frac{\partial^2}{\partial R^2} \phi(t_n, S_n, R_n) \\
    & \qquad \qquad \qquad + 12(S_n-\tilde S)^2 + 2\frac{\partial^2}{\partial S \partial R} \phi(t_n, S_n, R_n)  \bigg) \\
    & + \mathds 1_{\{\tilde y +\xi \le \overline Y\}}\bar{\lambda}^b(\tilde y, S_n) \big[ \tilde \beta^b(t_n, \tilde y,S_n) + \tilde u(t_n, \tilde y +\xi, S_n) - \tilde u(t_n, \tilde y, S_n) \big] \, \nonumber \\
    & + \mathds 1_{\{\tilde y -\xi \ge \underline Y\}} \bar{\lambda}^a(\tilde y, S_n) \big[ \tilde \beta^a(t_n, \tilde y,S_n) + \tilde u(t_n, \tilde y - \xi, S_n) - \tilde u(t_n, \tilde y, S_n)\big]\\
    & - \mathds 1_{\{\tilde y +\xi \le \overline Y\}}\bar{\lambda}^b(\tilde y, R_n) \big[ \tilde \beta^b(t_n, \tilde y,R_n) + \tilde w(t_n, \tilde y +\xi, R_n) - \tilde w(t_n, \tilde y, R_n) \big] \, \nonumber \\
    & - \mathds 1_{\{\tilde y -\xi \ge \underline Y\}} \bar{\lambda}^a(\tilde y, R_n) \big[ \tilde \beta^a(t_n, \tilde y,R_n) + \tilde w(t_n, \tilde y - \xi, R_n) - \tilde w(t_n, \tilde y, R_n) \big],
\end{align*}
where $C_n$ does not depend on $\eta$. Notice also that, by  \eqref{app:contrad0}, we have
\begin{align*}
    \underset{n\rightarrow \infty}{\lim \sup}\; \tilde u (t_n, \tilde y+\xi, S_n) - \tilde w(t_n, \tilde y+\xi, R_n) \le \tilde u (\tilde t, \tilde y+\xi, \tilde S) - \tilde w(\tilde t, \tilde y+\xi, \tilde S) \le  \tilde u (\tilde t, \tilde y, \tilde S) - \tilde w(\tilde t, \tilde y, \tilde S),
\end{align*}
and similarly
\begin{align*}
    \underset{n\rightarrow \infty}{\lim \sup}\; \tilde u (t_n, \tilde y-\xi, S_n) - \tilde w(t_n, \tilde y-\xi, R_n) \le \tilde u (\tilde t, \tilde y-\xi, \tilde S) - \tilde w(\tilde t, \tilde y-\xi, \tilde S) \le  \tilde u (\tilde t, \tilde y, \tilde S) - \tilde w(\tilde t, \tilde y, \tilde S).
\end{align*}
Therefore, sending $\eta$ to 0 and using the facts that $\tilde u, \tilde w \in \Xi$ and the functions $\tilde \beta^b, \tilde \beta^a, \bar \lambda^b, \bar \lambda^a$ are continuous, we can find a sequence $(\zeta_n)_n$ such that $\zeta_n \longrightarrow 0$ as $n \rightarrow \infty$ and
\begin{align*}
    \rho \left( \tilde u ( t_n,\tilde y,S_n) - \tilde w ( t_n,\tilde y,R_n)\right) \le &\; \zeta_n +\frac{\partial}{\partial t}\phi(t_n, S_n, R_n)  \\
    & + \frac 12 \sigma^2 \bigg(\frac{\partial^2}{\partial S^2} \phi(t_n, S_n, R_n) + \frac{\partial^2}{\partial R^2} \phi(t_n, S_n, R_n)  + 2\frac{\partial^2}{\partial S \partial R} \phi(t_n, S_n, R_n)  \bigg).
\end{align*} 
Now sending $n\to\infty$ yields
\begin{align*}
    \rho \left( \tilde u ( \tilde t,\tilde y,\tilde S) - \tilde w ( \tilde t,\tilde y,\tilde S)\right) \le &\;\frac{\partial}{\partial t}\phi(\tilde t, \tilde S, \tilde S) + \frac 12 \sigma^2 \bigg(\frac{\partial^2}{\partial S^2} \phi(\tilde t, \tilde S, \tilde S) + \frac{\partial^2}{\partial R^2} \phi(\tilde t, \tilde S, \tilde S)  + 2\frac{\partial^2}{\partial S \partial R} \phi(\tilde t, \tilde S, \tilde S)  \bigg).
\end{align*}
But it is clear that for $\mu>0$ large enough, the right-hand side of this inequality is negative, which contradicts  \eqref{app:contrad0} as $\rho>0$. 
\end{proof}

We conclude this section with the following result.

\begin{corollary}
    The value function $v$ is continuous. Moreover, it is the only viscosity solution to the HJB QVI \eqref{eq:hjb_qvi} in $\Xi$ with terminal condition $v(T,y,S) = 0$ for $(y,S) \in Q\times \mathbb R$.
\end{corollary}

\begin{proof}
    We know that $v$ is in $\Xi$ and is a viscosity solution to the HJB QVI \eqref{eq:hjb_qvi}. In particular, $v^*$ is a viscosity subsolution to \eqref{eq:hjb_qvi} and $v_*$ is a viscosity supersolution to $\eqref{eq:hjb_qvi}$. We also know from Propositon \ref{prop:terminal} that $v^*(T,y,S) = v_*(T,y,S) = 0$ for all $(y,S) \in Q\times \mathbb R$.\\

    Hence $v^*$ and $v_*$ verify the assumptions of Proposition \ref{prop:comparison} and we obtain $v_* \ge v^*$ on $\mfT\times Q \times \mathbb R$. But by definition we also have $v_*\le v\le v^*$ on $\mfT\times Q \times \mathbb R$, so we can conclude that $v_*= v= v^*$, and in particular $v$ is continuous.\\

    Let us now assume that we have an other viscosity solution $V \in \Xi$ to the HJB QVI \eqref{eq:hjb_qvi} with terminal condition $V(T,y,S) = 0$ for all $(y,S) \in Q\times \mathbb R$. Using the same reasoning, $V$ is necessarily continuous.\\

    Because $V$ is a subsolution to \eqref{eq:hjb_qvi} and $v$ is a supersolution to \eqref{eq:hjb_qvi}, and as $v(T,y,S) = V(T,y,S) = 0$ for all $(y,S) \in Q\times \mathbb R$, we know by Proposition \ref{prop:comparison} that $V \le v $ on $\mfT\times Q \times \mathbb R$. But we also have that $V$ is a supersolution and $v$ a subsolution to \eqref{eq:hjb_qvi}, which gives us $V \ge v $ on $\mfT\times Q \times \mathbb R$. Therefore, we obtain $V=v$, which proves the result.    
\end{proof}

\bibliographystyle{apalike}
\bibliography{References}

\section*{Acknowledgment}

The authors would like to thank Bruno Bouchard and Olivier Guéant for their insightful comments. 

\section*{Statement on Potential Conflicts of Interest}

The second author is a PhD student at Université Paris Dauphine–PSL, funded by the Research Initiative ``Blockchain et Intelligence Artificielle pour les infrastructures de marché,'' under the aegis of the Institut Europlace de Finance, in partnership with LCH SA (an LSEG business). The authors bear sole responsibility for the content of this publication, which does not reflect the views or practices of LCH SA.

\section*{Statement on Data Availability}

The data that support the findings of this study are available from the authors upon reasonable request.

\section*{Statement on ChatGPT Usage}

ChatGPT was used solely for language refinement and editing. All content, ideas, and analyses are original and remain the sole responsibility of the authors.

\section*{Statement on open access}
For the purpose of open access, the authors have applied a CC BY public copyright licence to any author accepted manuscript version arising from this submission.

\end{document}